\documentclass[aps,prl,amsmath,floats,floatfix,twocolumn,
  superscriptaddress,nofootinbib,showpacs]{revtex4-2}
\pdfoutput=1 
\usepackage{multirow}
\usepackage{amsmath}    
\usepackage{lipsum}
\usepackage{amssymb}
\usepackage{amsfonts}
\usepackage{amsthm}
\usepackage{braket}
\usepackage{bm}
\usepackage{multirow}
\usepackage{caption}
\usepackage{color} 
\usepackage{array}
\usepackage{caption}
\usepackage{soul}
\usepackage{enumitem}
\usepackage{multirow}
\usepackage{dcolumn}
\usepackage[dvipsnames]{xcolor}
\usepackage{epsfig}
\usepackage{graphicx}   
\usepackage{graphics}
\usepackage[latin1]{inputenc}
\usepackage{latexsym}
\usepackage{mathrsfs}
\usepackage{mathtools}
\usepackage{orcidlink}
\usepackage{rotating}
\usepackage{setspace}
\usepackage{siunitx}
\usepackage{subfigure}  
\usepackage{url}
\usepackage{verbatim}	
\usepackage{hyperref}   
\usepackage{float}
\hypersetup{
    colorlinks=true,
    urlcolor=blue,
    citecolor=blue,
    linkcolor=red
}
\definecolor{IsoGreen}{RGB}{0,120,70}
\definecolor{IsoRed}{RGB}{178,34,34}
\definecolor{HeaderGray}{RGB}{242,242,242}
\definecolor{SoftGreen}{RGB}{246,252,248}
\definecolor{SoftRed}{RGB}{255,247,247}

\usepackage{amsmath,amssymb,bm,mathtools}
\usepackage[T1]{fontenc}
\usepackage{lmodern}

\usepackage{amsthm}

\usepackage{xspace} 
\usepackage{cleveref}

\usepackage[toc,page]{appendix}
\usepackage{array}
\newtheorem{theoremA}{Theorem}
\newtheorem{theoremB}{Theorem}

\usepackage{yfonts}
\usepackage{tikz}
\usetikzlibrary{shapes.geometric, arrows}
\tikzstyle{null} = [rectangle, 
minimum width=0.5cm, 
minimum height=0.5cm, 
text centered, 
fill=white!30]
\tikzstyle{nullblue} = [rectangle, 
minimum width=0.5cm, 
minimum height=0.5cm, 
text centered, 
text = blue,
fill=white!30]
\tikzstyle{nullor} = [rectangle, 
minimum width=0.5cm, 
minimum height=0.5cm, 
text centered, 
text = orange,
fill=white!30]
\tikzstyle{nullred} = [rectangle, 
minimum width=0.5cm, 
minimum height=0.5cm, 
text centered, 
text = red,
fill=white!30]
\tikzstyle{nullsmall} = [rectangle, 
minimum width=0.1cm, 
minimum height=0.1cm]
\tikzstyle{bracket} = [rectangle, 
minimum width=0.3cm, 
minimum height=3cm]
\tikzstyle{ColinStyle} = [rectangle, 
minimum width=1.2cm, 
minimum height=1cm, 
text centered, 
draw=black, 
line width = 0.7pt,
fill=white!30]
\tikzstyle{arrow} = [line width = 1.3 pt,->,>=stealth]
\tikzstyle{line} = [line width = 0.9pt,-,>=stealth]
\tikzstyle{bluethickline} = [line width = 2pt,-,>=stealth, color = blue]
\tikzstyle{orthickline} = [line width = 2pt,-,>=stealth, color = orange]
\tikzstyle{redthickline} = [line width = 2pt,-,>=stealth, color = red]
\makeatletter
\newcommand*{\rom}[1]{\expandafter\@slowromancap\romannumeral #1@}

\newtheorem{lemmaA}{Lemma}

\newtheorem{lemmaB}{Lemma}

\definecolor{Purples}{HTML}{4B2E83}

\newtheorem{definition}{Definition}

\makeatother
\allowdisplaybreaks

\newcommand{\CAL}{Theoretical Astrophysics 350-17, California Institute of Technology, Pasadena, CA 91125, USA}
\newcommand{\ILL}{Illinois Center for Advanced Studies of the Universe \& Department of Physics, \\
University of Illinois Urbana-Champaign, Urbana, Illinois 61801, USA}
\newcommand{\AEI}{Max Planck Institute for Gravitational Physics (Albert Einstein Institute), D-14476 Potsdam, Germany}

\begin{document}
    \title{Chiral symmetry and black hole isospectrality}
    \author{Colin Weller\,\orcidlink{0000-0001-5173-5638}}
    \thanks{Contact author: \href{cweller@caltech.edu}{cweller@caltech.edu}}
    \affiliation{\CAL}
    
    \author{Dongjun Li\,\orcidlink{0000-0002-1962-680X}}
    \thanks{Contact author: \href{dongjun@illinois.edu}{dongjun@illinois.edu}}
    \affiliation{\ILL}

    \author{Pratik Wagle\,\orcidlink{0000-0003-3700-4227}}\affiliation{\AEI}

    \author{Andrew Laeuger\, \orcidlink{0000-0002-8212-6496}} \affiliation{\CAL}
    
    \author{Yanbei Chen\,\orcidlink{0000-0002-9730-9463}}
    \affiliation{\CAL}
    
    \author{Nicol\'{a}s Yunes\,\orcidlink{0000-0001-6147-1736}}
    \affiliation{\ILL}
    
    \date{\today}
    \begin{abstract}
We prove that black-hole isospectrality between linearly independent solutions follows whenever radiative metric and matter perturbations are reconstructed from complex master variables obeying a closed, complex-linear system with complex-linear boundary conditions. If a phase rotation maps one parity sector to the other, then the even- and odd-parity QNM spectra coincide, and the perturbations are parity isospectral. A broad class of chiral-aligned theories satisfies these criteria. As an application, we demonstrate that subextremal Kerr--de Sitter black holes are parity isospectral.
    \end{abstract}
    \maketitle

\noindent\textit{Introduction}---At the end of binary black hole coalescence, gravitational waves enter the ringdown phase, which is dominated by a superposition of quasinormal modes (QNMs) at late times~ \cite{Vishveshwara:1970zz,Press:1971wr, Davis:1971kk}. 
Each mode has a characteristic complex-valued frequency, which can be identified as a pole in the retarded Green's function of  Einstein's linearized equations in the frequency domain. 
Since stationary, asymptotically-flat, vacuum black holes within general relativity (GR) belong solely to the Kerr family \cite{Israel:1967wq, Carter:1971zc, Price:1971fb, Robinson:1975bv}, the QNM spectrum depends only on the mass and angular momentum. This underlies tests of GR with black hole spectroscopy \cite{Cardoso:2016rao,Berti:2016lat,Isi:2019aib,Isi:2020tac,Mitman:2025hgy,Berti:2025hly,Franchini:2025csk}, in which the measured QNM spectrum is used to constrain the properties of the remnant black hole and its environment. 

For nonrotating, Schwarzschild black holes in GR, gravitational perturbations are usually computed using the Regge--Wheeler--Zerilli formalism \cite{Regge:1957td, Zerilli:1970se,Zerilli:1970wzz, Moncrief:1974am}. In the proper gauge, the Einstein equations decouple into two second-order, one-dimensional, Schr\"odinger-type wave equations for even- and odd-parity perturbations, respectively \cite{Gerlach:1979rw, Martel:2005ir}. While the effective radial potentials, angular emission patterns, and excitation factors for the two parity sectors are distinct, Chandrasekhar and Detweiler showed that their spectra coincide \cite{Chandrasekhar:1975zza, Chandrasekhar_1983}, a
degeneracy known as \textit{parity isospectrality}. This property has been established for nonrotating black holes by relating the two parity potentials and the master functions through a Chandrasekhar transformation \cite{Chandrasekhar_1983}, a special case of a Darboux \cite{Darboux1882, Glampedakis:2017rar, trachanas2009exactly} or intertwining \cite{Anderson:1991kx} relation, which also has applications in supersymmetric quantum mechanics \cite{Golfand:1971iw, Witten:1981nf, Cooper:1982dm, Cooper:2001weo, FernandezC:2018cdo} and more generically in studies of isospectral Hamiltonians \cite{Pursey:1986kk}. Viewed more broadly, isospectrality can also be understood more generally as a frequency
degeneracy between linearly independent solutions, without reference to the parity of those solutions. 
We therefore use \textit{isospectrality} for this more general degeneracy between linearly independent modes, 
and reserve \textit{parity isospectrality} for the standard definition
that the even- and odd-parity spectra coincide.

Rotating black holes require a different formalism, since the loss of spherical symmetry prevents the two parity sectors from decoupling. To address this difficulty, Teukolsky formulated a curvature-based approach using the spinor calculus of Newman and Penrose \cite{Newman:1961qr,Penrose:1960eq}. For Ricci-flat, Petrov type D spacetimes, Teukolsky found two wave equations governing the radiative Weyl scalars $\Psi_0$ and $\Psi_4$, respectively \cite{Teukolsky:1972my, Teukolsky:1973ha, Press:1973zz, Teukolsky:1974yv}. These two scalars provide a gauge-invariant description of gravitational radiation from a perturbed black hole, from which one can further construct definite-parity modes \cite{Chrzanowski:1975wv, Nichols:2012jn}. One can show semi-analytically~\cite{Pani:2013ija, Pani:2013wsa, Franchini:2022axs} that vacuum perturbations of slowly-rotating Kerr black holes are parity isospectral. By exploiting the conjugate-parity symmetry of Teukolsky equations, Ref.~\cite{Li:2023ulk} further proved parity isospectrality for perturbations of Kerr black holes with arbitrary spin~\cite{Li:2023ulk}. This property is also supported by numerical investigations of the Kerr spectrum \cite{ Detweiler:1980gk, Leaver:1985ax}.

More recently, the Teukolsky formalism has been extended to enable fully-relativistic computations of the leading-order radiative corrections to gravitational waves in effective-field-theory extensions of GR \cite{Li:2022pcy, Hussain:2022ins, Cano:2023tmv, Cano:2024jkd, Wagle:2023fwl, Li:2025fci, Li:2025ffh}. This extension has enabled studies of isospectrality for rapidly-rotating black holes beyond GR \cite{Li:2023ulk}, demonstrating that isospectrality is an easy symmetry to lose. For instance, isospectrality is not present in theories with higher-derivative corrections \cite{Cardoso:2018ptl,Cardoso:2019mqo,Cano:2023jbk, Cano:2024ezp}, with nonminimally coupled scalar fields \cite{Li:2023ulk, Li:2025fci, Chung:2024ira}, with extra dimensions \cite{Cardoso:2002pa, Kodama:2003jz}, in anti-de Sitter spacetimes in GR \cite{Cardoso:2001bb, Berti:2003ud}, in beyond-Kerr geometries in GR \cite{Weller:2024qvo, Wu:2025obg}, in other exotic compact objects, such as gravastars \cite{Saketh:2024ojw, Pani:2009ss}, and in matter environments, such as in thin shells \cite{Laeuger:2025zgb} and for relativistic stars \cite{Thorne:1968zz, Kokkotas:1999bd}.

The above examples motivate a more fundamental question: what structures or
symmetries are sufficient to guarantee isospectrality and, more narrowly,
parity isospectrality?
Structural answers have recently emerged along two complementary
directions, one restricted by the background geometry and the other by
the multipolar regime. For static, spherically symmetric backgrounds,
Refs.~\cite{Mukkamala:2024dxf,Pereniguez:2026avs} have proposed that
parity isospectrality is organized by closed equations for self-dual curvature
variables, whose real and imaginary parts encode the even- and odd-parity
sectors, respectively. In the eikonal limit, where the multipole number
$\ell$ is much greater than unity, parity isospectrality has been conjectured to
be equivalent to non-birefringent gravitational-wave propagation
\cite{Silva:2019scu,Cano:2024wzo} and, more recently, to gravitational
electric--magnetic duality at the light ring \cite{Bah:2026aia}.
Together, these results point toward a common complex, or duality,
structure underlying isospectrality. Observable ringdown signals,
however, are typically dominated by low radiative multipoles of rapidly
rotating black holes. The structure underlying isospectrality in this
regime, with or without matter or a cosmological constant, has yet to be
identified.

In this Letter, we identify such a structure and formulate it as a sufficient
criterion in Theorems~\ref{thm: MainTheorem} and \ref{thm: TheoremB}. 
Our results require no expansion in the spin, the inverse multipole number, or the small couplings of an
effective field theory. We further show that \textit{chiral-aligned} theories form a broad class that satisfy 
the conditions of Theorem~\ref{thm: MainTheorem}. Because the conditions of the theorem are sufficient rather than necessary, 
a failure to satisfy them does not imply that isospectrality is broken. With this in mind, we assess the status of isospectrality 
within and beyond GR for a variety of black holes, summarized in Table~\ref{tab:isospectrality_examples_reorganized}. In particular, we establish for the first time that linear perturbations of the Kerr--de~Sitter black hole are parity isospectral at arbitrary subextremal spin. This contrasts with Kerr anti--de~Sitter black holes, for which parity 
isospectrality is broken~\cite{Cardoso:2001bb,Berti:2003ud}.


\vspace{0.2cm}
\noindent \textit{Formalism and isospectrality theorems}---Let $S$ be a real-valued action that is a functional of a metric
$g_{\mu\nu}$ (where Greek letters stand for spacetime indices), assumed to have a mostly positive, Lorentzian signature,
and a collection of additional scalar, vector and tensor fields $\varphi_{\mathcal{A}}$ of the form 
\begin{align}\label{eq: Action}
    S[g_{\mu \nu},\varphi_{\mathcal{A}}]=\int d^4 x \sqrt{|g|}\left(\kappa_g R+  
    \mathcal{L}\left(\varphi_{\mathcal{A}}, g_{\mu \nu}, \nabla \varphi_{\mathcal{A}},\dots\right)\right)\,,
\end{align} 
where $\kappa_g = 1/(16\pi)$ in geometric ($G=1=c$) units, and $\mathcal{L}(\varphi_{\mathcal{A}},g_{\mu \nu},\nabla \varphi_{\mathcal{A}},\dots)$ is a generic local Lagrangian density invariant under both coordinate diffeomorphisms and Lorentz transformations. Varying the action with respect to the metric and the matter fields yields 
\begin{subequations}\label{eq: EquationsOfMotion}
\begin{align}
    & G_{\mu \nu}=\frac{1}{2 \kappa_g} T_{\mu \nu}^{\text{eff}}\,,\;
    && \quad T_{\mu \nu}^{\text{eff}} \equiv \frac{-2}{\sqrt{|g|}} \frac{\delta(\sqrt{|g|} \mathcal{L})}{\delta g^{\mu \nu}}\,, \\
    & \mathcal{M}[\varphi_{\mathcal{A}}, g_{\mu\nu}]= 0\,,\;
    && \mathcal{M}[\varphi_{\mathcal{A}}, g_{\mu\nu}] \equiv \frac{\delta S}{\delta \varphi_{\mathcal{A}}}\,.
    \label{eq: MatterEquOfMotion}
\end{align}
\end{subequations}
where $G_{\mu \nu}$ is the Einstein tensor, $T_{\mu \nu}^{\text {eff }}$ is the effective stress-energy tensor produced from the variation of $\mathcal{L}(\varphi_{\mathcal{A}},g_{\mu \nu},\nabla \varphi_{\mathcal{A}}...)$, and $\mathcal{M}[\varphi_{\mathcal{A}}, g_{\mu \nu}]$ denotes the system of equations in the matter sector. Let the pair $\left(g_{\mu \nu}, \varphi_{\mathcal{A}}\right)$ be a stationary background solution of the equations of motion in Eq.~\eqref{eq: EquationsOfMotion}, and let $(g_{\mu \nu } + \epsilon h^{(1)}_{\mu \nu}, \varphi_{\mathcal{A}} + \epsilon \varphi^{(1)}_{\mathcal{A}})$ be a linearized solution, which jointly solve Eq.~\eqref{eq: EquationsOfMotion} to $\mathcal{O}(\epsilon^1)$ with $\epsilon\ll 1$ a bookkeeping perturbation parameter. First-order quantities at $\mathcal{O}(\epsilon^1)$ will receive a superscript i.e. $\mathcal{Q}^{(1)}$, while background quantities at $\mathcal{O}(\epsilon^0)$ will have no superscript.

\begin{theoremA} \label{thm: MainTheorem}
Suppose that, for every linear solution $(g_{\mu \nu} + \epsilon h^{(1)}_{\mu \nu},\varphi_{\mathcal{A}} + \epsilon \varphi^{(1)}_{\mathcal{A}})$ of Eq.~\eqref{eq: EquationsOfMotion}, the radiative components can be entirely reconstructed, up to a coordinate-gauge transformation and time-translation, from a set of complex-valued master variables $\{f^{(1)}_i\}$. Specifically, there are reconstruction maps ${}_{h}C_{\mu \nu}^i$ and ${}_{\varphi_{\mathcal{A}}}C^i$ such that\footnote{The reconstruction in  Eq.~\eqref{eq: ReconstructMetricandMatter} is performed in the time domain, which is possible since our background spacetime is stationary.} 
\begin{subequations}\label{eq: ReconstructMetricandMatter}
\begin{align}
    h^{(1)}_{\mu \nu} &= {}_{h}C_{\mu \nu}^i f^{(1)}_i + {}_{h}\bar{C}_{\mu \nu}^i \bar{f}^{(1)}_i\,,\\
     \varphi^{(1)}_{\mathcal{A}} &= {}_{\varphi_{\mathcal{A}}}C^i f^{(1)}_i + {}_{\varphi_{\mathcal{A}}}\bar{C}^i \bar{f}^{(1)}_i\,,
\end{align}
\end{subequations}
where we sum over the internal index $i$. Furthermore, we assume that the reconstruction map of Eq.~\eqref{eq: ReconstructMetricandMatter} is a linear isomorphism from the complex master variables to the real-valued radiative solutions: every radiative solution is linearly reconstructed by a unique set of master variables $\{f^{(1)}_i\}$.
If the linearized equations of motion reduce to a complex-linear\footnote{Here, we take complex-linear to mean linear over the complex numbers. For instance, a system of equations $L f_i^{(1)} = 0$ is complex-linear if $L(c f^{(1)}_i) = c\, L f^{(1)}_i$ for any $c \in \mathbb{C}$.} system of differential equations in $\{f^{(1)}_i\}$, and the imposed boundary conditions are  complex-linear and define a QNM, then the linear perturbations are isospectral. 
\end{theoremA}

\begin{proof}
Let $\{f_i\}$ be a nonzero set of master variables satisfying the hypotheses of Theorem \ref{thm: MainTheorem}. Since the equations of motion and boundary conditions for $\{f_i\}$ are complex-linear and each master variable is complex-valued, the set $e^{i \theta} \{f_i\}$ for a constant phase $\theta$ in $(0,\pi)$ also forms a set of master variables satisfying the hypotheses of Theorem \ref{thm: MainTheorem}. 

Suppose $\{f_i\}$ and $e^{i \theta}\{f_i\}$ reconstruct the linear solutions
$(h^{(1)}_{\mu \nu}, \varphi^{(1)}_{\mathcal{A}})$ and
$(\tilde{h}^{(1)}_{\mu \nu}, \tilde{\varphi}^{(1)}_{\mathcal{A}})$, respectively,
via Eq.~\eqref{eq: ReconstructMetricandMatter}. We now show that these two
real-valued solutions are linearly independent. First, the
reconstruction in Eq.~\eqref{eq: ReconstructMetricandMatter} is real-linear: for real constants $a$ and
$b$, the set $\{g_i\} \equiv a \{f_i\} + b\, e^{i \theta}\{f_i\}$ is another set of complex master variables which reconstructs
the real-valued fields $(a \, h^{(1)}_{\mu \nu} + b \, \tilde{h}^{(1)}_{\mu \nu},\;
a \, \varphi^{(1)}_{\mathcal{A}} + b \, \tilde{\varphi}^{(1)}_{\mathcal{A}})$. Now
suppose the two solutions were proportional to each other, i.e.
$(\tilde{h}^{(1)}_{\mu \nu}, \tilde{\varphi}^{(1)}_{\mathcal{A}})
= (c\, h^{(1)}_{\mu \nu}, c\, \varphi^{(1)}_{\mathcal{A}})$ for some real
constant $c$. Since the reconstruction is real-linear and there is an
isomorphism between the linearized solutions and the master variables, the
master variables would then satisfy $e^{i \theta}\{f_i\} = c \{f_i\}$, which
requires $e^{i\theta}$ to be real. This is a contradiction for
$\theta$ in $(0,\pi)$, and therefore the two solutions are linearly independent.

Both solutions $(h^{(1)}_{\mu \nu}, \varphi^{(1)}_{\mathcal{A}})$ and
$(\tilde{h}^{(1)}_{\mu \nu}, \tilde{\varphi}^{(1)}_{\mathcal{A}})$ are
reconstructed from master variables that satisfy the same equations of motion
and the same boundary conditions, and therefore share the same QNM spectrum.
Thus, every QNM solution admits a linearly independent counterpart that is
itself a QNM solution with the same complex frequency. $\square$
\end{proof} 

\vspace{0.1cm}
Now we formulate a special case of Theorem \ref{thm: MainTheorem} where the phase rotation of the master variables corresponds to a combination of solutions from opposite parity sectors, recovering the standard definition of parity isospectrality. We define parity following Ref.~\cite{Li:2023ulk} where we consider the background spacetime and matter configuration in the Boyer--Lindquist coordinates. Then, for scalar functions $f(t,r,\theta,\phi)$ in Boyer--Lindquist coordinates, the parity operator $\hat{P}$ acts as $\hat{P}f(t,r,\theta,\phi) = f(t,r,\pi - \theta, \phi + \pi)$. We also fix our frame to be the canonical Boyer-Lindquist frame, which is aligned with the angular momentum of the remnant black hole. For linear perturbations, definite-parity modes must satisfy $\hat{P}h^{(1)} =\pm (-1)^{\ell}h^{(1)}$ (with $\pm$ referring to even/odd modes respectively), where $h^{(1)}$ is a generic first-order quantity and $\ell$ is the angular momentum number, after decomposing into angular harmonics \cite{Regge:1957td,Zerilli:1970se,Moncrief:1974am, Nichols:2012jn,Li:2023ulk}.

\begin{theoremB}\label{thm: TheoremB}
Under the hypotheses of Theorem \ref{thm: MainTheorem}, suppose further that,
whenever a nonzero set of master variables $\{f_i\}$ reconstructs a
definite-parity solution, there exists a phase-rotated set of master variables
$e^{i\theta}\{f_i\}$, for some real-valued constant $\theta$ in the open set $(0,\pi)$, which reconstructs an opposite definite-parity
solution. Then, the linear perturbations are
parity isospectral.
\end{theoremB}
\begin{proof}
    We follow the proof of Theorem \ref{thm: MainTheorem}. Let $\{f_i\}$ reconstruct a definite-parity solution $(h^{(1)}_{\mu \nu}, \varphi^{(1)}_{\mathcal{A}})$.
    By complex-linearity of the
equations of motion and of the boundary conditions, $e^{i\theta}\{f_i\}$ is
again a set of complex master variables, and by
Eq.~\eqref{eq: ReconstructMetricandMatter}, it reconstructs a solution
$(\tilde{h}^{(1)}_{\mu\nu}, \tilde{\varphi}^{(1)}_{\mathcal{A}})$ at the same frequency as $(h^{(1)}_{\mu \nu}, \varphi^{(1)}_{\mathcal{A}})$. By hypothesis,  $(\tilde{h}^{(1)}_{\mu\nu}, \tilde{\varphi}^{(1)}_{\mathcal{A}})$ must be the opposite parity of $(h^{(1)}_{\mu \nu}, \varphi^{(1)}_{\mathcal{A}})$. Hence every QNM frequency
of one parity sector also belongs to the spectrum of the other, establishing
parity isospectrality. $\square$
\end{proof}

\vspace{0.2cm}
\noindent \textit{Chiral-aligned theories}---In Theorem \ref{thm: MainTheorem} we established a sufficient, but not necessary, criterion for two linearly independent solutions to have the same QNM frequency. Parity isospectrality is then recovered as a limiting case in Theorem \ref{thm: TheoremB}, where a phase rotation of the definite-parity master variables corresponds to a parity transformation.  One broad class of theories, which we refer to as chiral-aligned, manifestly satisfies all the conditions in Theorem~\ref{thm: MainTheorem} as long as the boundary conditions for the perturbations remain complex-linear, specifically:
\begin{definition}[Chiral-aligned theory]\label{def: ChiralAligned}
Consider a theory of the form of Eq.~\eqref{eq: Action}, together with a background
solution of its field equations [Eq.~\eqref{eq: EquationsOfMotion}], which is compatible
with a null tetrad $\{l^{\mu},n^{\mu},m^{\mu},\bar{m}^{\mu}\}$, 
i.e.~$g_{\mu \nu} = -2 l_{(\mu} n_{\nu)} + 2 m_{(\mu} \bar{m}_{\nu)}$.
We say such a theory is chiral-aligned if it
satisfies the following five conditions:

\begin{enumerate}[label=(\roman*)]
    \item The background spacetime $g_{\mu\nu}$ is conformally K\"ahler\footnote{Here we use the Lorentzian analogue of conformally K\"ahler geometry \cite{Aksteiner:2022bwr}.}
    [see Eq.~\eqref{eq: CKCondition}]\,;
    \item If the background spacetime is not Ricci flat, it satisfies the
    matter-aligned condition [see Eq.~\eqref{eq: RicciRequire}]\,;
    \item 
    The following perturbed
    components of the effective stress-energy tensor vanish:
    $T_{\mu \nu}^{\text{eff}(1)}l^{\mu}l^{\nu} = 0 =T_{\mu \nu}^{\text{eff}(1)}n^{\mu}n^{\nu}$
    ;
    \item When the equations of motion are imposed, there exists a gauge in
    which the perturbed Newman--Penrose Ricci scalars satisfy\footnote{Henceforth, 
    we use tetrad shorthand notation, so that, e.g.,~$h_{n m} = h_{\mu \nu} n^{\mu} m^{\nu}$.}
    \begin{subequations}\label{eq: RicciCL}
    \begin{align}
    \Phi^{(1)}_{01} &=
    \mathcal{S}_{01}[h^{(1)}_{nm},h^{(1)}_{lm},h^{(1)}_{mm}]\,, \\
    \Phi^{(1)}_{02} &=
    \mathcal{S}_{02}[h^{(1)}_{nm},h^{(1)}_{lm},h^{(1)}_{mm}]\,, \\
    \Phi^{(1)}_{21} &=
    \mathcal{S}_{21}[h^{(1)}_{n\bar{m}},h^{(1)}_{l\bar{m}},h^{(1)}_{\bar{m}\bar{m}}]\,, \\
    \Phi^{(1)}_{20} &=
    \mathcal{S}_{20}[h^{(1)}_{n\bar{m}},h^{(1)}_{l\bar{m}},h^{(1)}_{\bar{m}\bar{m}}]\,,
    \end{align}
    \end{subequations}
    where the operators $\mathcal{S}_{ij}$ are theory-dependent, complex-linear
    maps of the metric components, possibly containing derivatives, whose
    coefficients depend solely on background quantities;
    \item Once the metric perturbation $h^{(1)}_{\mu \nu}$ has been reconstructed, the matter 
perturbations $\varphi^{(1)}_{\mathcal{A}}$ can be reconstructed, up to a gauge and time translation, from the \textit{same} master variables used to reconstruct the metric and from the reconstructed metric. In particular, the matter perturbations do not satisfy
a decoupled set of evolution equations that would require additional master functions to reconstruct. 
\end{enumerate}
\end{definition}
\begin{lemmaA}\label{lem: ChiralLemma}
If the boundary conditions are complex-linear, the linear perturbations of a chiral-aligned theory are isospectral. 
\end{lemmaA} 

\begin{proof}
We begin by showing how conditions (i) and (ii) of Def.~\ref{def: ChiralAligned} admit a convenient frame gauge for background quantities. 

Firstly, condition (i) requires the spacetime is conformally K\"ahler. An orientable four-manifold $(M,g)$ is conformally K\"ahler if and only if it admits a non-null symmetric spinor field $\chi_{AB}$ satisfying \cite{Dunajski:2009dqa}
\begin{align}\label{eq: CKCondition}
\nabla_{A'(A}\chi_{BC)}=0\,.
\end{align}
The spinor indices $(A,B,\dots)$ are raised and lowered with the antisymmetric Levi-Civita spinors $\epsilon_{AB}$ and
$\epsilon_{A'B'}$, and a spacetime index is equivalent to a pair of spinor indices, $\mu \leftrightarrow AA'$.
We work with a normalized spinor dyad $(o^A,\iota^A)$ obeying $o_A\iota^A = 1$, which is the spinorial counterpart of the Newman-Penrose null tetrad
$(l^\mu,n^\mu,m^\mu,\bar m^\mu)$ via
$l^\mu\leftrightarrow o^A\bar{o}^{A'}$, 
$n^\mu\leftrightarrow\iota^A\bar{\iota}^{A'}$, 
$m^\mu\leftrightarrow o^A\bar{\iota}^{A'}$, and
$\bar m^\mu\leftrightarrow\iota^A\bar{o}^{A'}$.

Additionally, conformally K\"ahler spacetimes satisfy an integrability constraint $\Psi_{(A B C}{ }^E \chi_{D) E}=0$ \cite{Andersson:2015xla}, 
which requires the self-dual Weyl tensor to take the form
\begin{align}\label{eq: CKRequire}
    \Psi_{A B C D} & = 6 \, \Psi_2 \; o_{(A} o_B \iota_C \iota_{D)}\,,
\end{align}
imposing the spacetime is Petrov type D \cite{PETROV196988, Stephani:2003tm, Petrov:2000bs}. Here, $\Psi_2$ is a Newman-Penrose
scalar, and we henceforth adopt standard definitions for all other Newman--Penrose quantities~\cite{Chandrasekhar:1975nkd, Li:2022pcy}.
For example, a static, spherically symmetric black hole is always conformally K\"ahler \cite{Aksteiner:2022bwr}. Rotating black holes, such as Kerr, are also conformally K\"ahler \cite{Aksteiner:2022bwr, Araneda:2025lak}. 

Secondly, condition (ii) imposes the background solution is a \textit{matter-aligned} spacetime, which satisfies
\begin{align}\label{eq: RicciRequire}
    \Phi_{(A}{ }^C{ }_{\left|A^{\prime} B^{\prime}\right|} \chi_{B) C} = 0\,,
\end{align}
where $\Phi_{AA'BB'}$ is the traceless Ricci spinor.
While Eq.~\eqref{eq: CKRequire} enforces the background to be Petrov type D, Eq.~\eqref{eq: RicciRequire} enforces the Ricci spinor to only have the Coulomb-type component $\Phi_{11}$ along with a trace $\Lambda_{\mathrm{NP}}$\footnote{ We use the notation $\Lambda_{\mathrm{NP}}$ for the Newman--Penrose Ricci scalar trace to differentiate it from the cosmological constant $\Lambda$ which we denote with no subscript.}. Importantly, if we dropped the requirement of Eq.~\eqref{eq: RicciRequire}, the Ricci spinor can be more general. For example, the static, electrically-charged black holes of Gibbons and Maeda \cite{Gibbons:1987ps} in Einstein-Maxwell-dilaton-axion gravity \cite{Pope:2025jgz} are still conformally K\"ahler but have nonzero Ricci scalars $\Phi_{00}$ and $\Phi_{22}$ on the background.

In a Newman--Penrose tetrad aligned with the repeated principal null directions, spacetimes satisfying Eqs.~\eqref{eq: CKCondition} and \eqref{eq: RicciRequire} obey the following background relations:
\begin{subequations}\label{eq: CKFrameGauge}
\begin{align}
    &\Psi_0=\Psi_1=\Psi_3=\Psi_4=0\label{eq: WeylScalarsVanish}\,,\\
    &\kappa=\sigma=\lambda=\nu=0\label{eq: GSEqu}\,,\\
    &\Phi_{ij}=0\quad\text{for}\quad(i,j)\neq(1,1)\,,\label{eq: ColRicci}
      \\ 
    &\Phi_{11}\ \text{and}\ \Lambda_{\mathrm{NP}}
      \ \text{may be nonzero}\,.\label{eq: ColRicciNonzero}
\end{align}
\end{subequations}
where Eq.~\eqref{eq: WeylScalarsVanish} is a consequence of Eq.~\eqref{eq: CKRequire}, Eq.~\eqref{eq: GSEqu} results from the Goldberg-Sachs theorem \cite{Goldberg2009}, and Eqs.~\eqref{eq: ColRicci} and \eqref{eq: ColRicciNonzero} are a consequence of Eq.~\eqref{eq: RicciRequire}. For example, 
type-D black holes in electrovacuum, such as Kerr--Newman or, more generally, the Pleba\'nski--Demia\'nski spacetimes \cite{PLEBANSKI197698}, are matter-aligned. 

Next, we will show how conditions (iii) and (iv) imply the radiative components of the metric and matter perturbations can be reconstructed from a closed set of complex-linear equations.  We consider the two Bianchi identities used to derive the Teukolsky equation for $\Psi^{(1)}_0$ \cite{Teukolsky:1973ha}: 
\begin{subequations}\label{eq: RicciFormCurved}
    \begin{align}
        F_1 \Psi^{(1)}_0 \!+\! F_2 \Psi^{(1)}_1   \!+\!\kappa^{(1)}\!\left( 3 \Psi^{(0)}_2 + 2 \Phi^{(0)}_{11}\right)&\!=\! \mathcal{S}_1\Phi^{(1)}_{00} \!+\! \mathcal{S}_2\Phi^{(1)}_{01}  \label{eq: Bianchi1}\,,\\
        F_3 \Psi^{(1)}_0 \!+\!F_4 \Psi^{(1)}_1   \!+\! \sigma^{(1)}\!\left( 3\Psi^{(0)}_2 - 2 \Phi^{(0)}_{11}\right) &\!= \mathcal{S}_3\Phi^{(1)}_{01} \!+\! \mathcal{S}_4\Phi^{(1)}_{02}\,,\label{eq: Bianchi2}
    \end{align}
\end{subequations}
where we have used the background tetrad relations for a matter-aligned, conformally K\"ahler spacetime given in Eq.~\eqref{eq: CKFrameGauge} 
(i.e., conditions (i) and (ii) in  Def.~\ref{def: ChiralAligned}). The first-order differential operators $\mathcal{S}_{i}$ and $F_i$ are given in the Supplemental Material \cite{SupplementalMaterial}. The first Bianchi identity in Eq.~\eqref{eq: RicciFormCurved} contains the real Ricci scalar $\Phi_{00}^{(1)}$. When expressed in terms of the complex metric variables, a nonvanishing $\Phi_{00}^{(1)}$ generally depends on both the chiral metric components and their complex conjugates, thereby preventing the equations for $\{h^{(1)}_{nm},h^{(1)}_{mm}\}$ from forming a closed, complex-linear system. In the radiation and tetrad gauges introduced below, however, condition (iii) and the $ll$ projection of the linearized field equations imply $\Phi_{00}^{(1)}\propto T_{ll}^{\mathrm{eff}(1)}=0$. Its vanishing is therefore an on-shell consequence of the field equations under condition (iii), rather than an off-shell geometric identity.

Now, let us show that all the gravitational fields can be reconstructed from a subset of metric fields satisfying complex-linear equations. Diffeomorphism invariance of the action [Eq.~\eqref{eq: Action}] and the vanishing of $\Phi^{(1)}_{00}$ allow one to impose the ingoing radiation gauge on the metric perturbation \cite{Price:2006ke}. The ingoing radiation gauge is traceless and satisfies $h_{\mu\nu}l^{\mu}=0$ where $l^{\mu}$ is the outgoing principal null direction of the background. In this gauge, the independent nonvanishing metric components may be taken to be $\{h^{(1)}_{nn},h^{(1)}_{nm},h^{(1)}_{mm}\}$, which comprise five real unknowns. Upon linearizing the Ricci and Bianchi identities in the Newman--Penrose formalism, the perturbed spin coefficients and Weyl scalars can be expressed in terms of derivatives of the perturbed metric. At $\mathcal{O}(\epsilon^1)$, we adopt the tetrad choice in Eq.~\eqref{eq: CCKFrameIRG}. Under this choice, the quantities entering Eq.~\eqref{eq: RicciFormCurved} can be written entirely in terms of $h^{(1)}_{nm}$ and $h^{(1)}_{mm}$ [see Eqs.~\eqref{eq: ReconstructedIRGQuantsReduced} and \eqref{eq: Psi_0MetricComp}].

After imposing the ingoing radiation gauge and adopting this tetrad choice, condition (iv) in Def.~\ref{def: ChiralAligned} ensures that $\Phi^{(1)}_{01}$ and $\Phi^{(1)}_{02}$ depend only on derivatives of $h^{(1)}_{nm}$ and $h^{(1)}_{mm}$. Consequently, the two Bianchi identities in Eq.~\eqref{eq: RicciFormCurved} form a closed, complex-linear system of two third-order differential equations for $h^{(1)}_{nm}$ and $h^{(1)}_{mm}$. Once this system has been solved, the remaining metric component $h^{(1)}_{nn}$ can be obtained from the Ricci identity for $\Phi^{(1)}_{22}$ [see Eq.~\eqref{eq: RiccihnnPerturb} of the Supplemental Material]. 

Once $h^{(1)}_{nm}$ and $h^{(1)}_{mm}$ have been computed, the remaining radiative component $h^{(1)}_{nn}$ is fixed. Consider two metric perturbations $h^{(1)}_{\mu\nu}$ and $q^{(1)}_{\mu\nu}$ satisfying $h^{(1)}_{nm}=q^{(1)}_{nm}$ and $h^{(1)}_{mm}=q^{(1)}_{mm}$. The difference metric, $j^{(1)}_{\mu\nu}\equiv h^{(1)}_{\mu\nu}-q^{(1)}_{\mu\nu}$, therefore satisfies $j^{(1)}_{nm}=j^{(1)}_{mm}=0$ and is also in the ingoing radiation gauge. The corresponding Weyl scalar obeys $\Psi^{(1)}_0[j]=0$. This implies, $j^{(1)}_{\mu\nu}$ can then contain only nonradiative components \cite{Wald:1973wwa}, demonstrating the radiative part of $h^{(1)}_{nn}$ is uniquely determined from $h^{(1)}_{nm}$ and $h^{(1)}_{mm}$. 

Additionally, since $h_{mm}$ and $h_{nm}$ are generically complex-valued, the set of components $\{\tilde{h}^{(1)}_{nm},\tilde{h}^{(1)}_{mm}\} = e^{i\theta}\{h^{(1)}_{nm},h^{(1)}_{mm}\}$ are not a real multiple of $\{h^{(1)}_{nm},h^{(1)}_{mm}\}$. The two sets of variables therefore reconstruct real-valued solutions $\tilde{h}^{(1)}_{\mu \nu }$ and $h^{(1)}_{\mu \nu}$ which are linearly independent. 
Since we have reconstructed the entire metric perturbation, we can reconstruct the matter $\varphi^{(1)}_{\mathcal{A}}$ by condition (v).  Therefore, the criterion of Theorem \ref{thm: MainTheorem} is satisfied, and the linear perturbations of a chiral-aligned theory are isospectral. $\square$
\end{proof}

The above construction shows that, for chiral-aligned theories, the criterion of Theorem \ref{thm: MainTheorem} are realized with the master functions given directly by the metric components $\{h^{(1)}_{nm},h^{(1)}_{mm}\}$. Importantly, isospectrality is not gauge dependent as QNM frequencies are observable in a detector response. For instance, a set of equations analogous to Eq.~\eqref{eq: RicciFormCurved} can also be obtained in the outgoing radiation gauge, a traceless gauge where $h_{\mu \nu}n^{\mu} = 0$ and $g_{\mu \nu}l^{\mu}n^{\nu} = -1$, from the two Bianchi identities used to derive the Teukolsky equation for $\Psi^{(1)}_4$. In this case, the master functions are now the set of metric components $\{h_{l\bar{m}},h_{\bar{m}\bar{m}}\}$. Here reconstruction is built into the formulation: the Bianchi identities in Eq.~\eqref{eq: RicciFormCurved} close on $\{h^{(1)}_{nm},h^{(1)}_{mm}\}$, and the only remaining component, $h^{(1)}_{nn}$, follows from the Ricci identity for $\Phi^{(1)}_{22}$. We additionally highlight that chiral-aligned theories have two master functions, but generically Theorems \ref{thm: MainTheorem} and \ref{thm: TheoremB} have no such restriction, as it places no limit on the number of master variables.

Next, we form an additional lemma which is a sufficient criterion to guarantee parity isospectrality:
\begin{lemmaB}\label{lem: ChiralLemmaParity}
Let the hypotheses of Lemma~\ref{lem: ChiralLemma} hold. Suppose in
addition that the null tetrad $\{l^{\mu},n^{\mu},m^{\mu},\bar{m}^{\mu}\}$ of the background spacetime obeys the following relations: $\hat{P} l^{\mu} = l^{\mu}, \hat{P} n^{\mu} = n^{\mu}$ and $\hat{P}m^{\mu} = e^{i \chi} \bar{m}^{\mu}$ where $\chi$ is a constant phase in $[0,\pi]$ and $\hat{P}$ is the parity operator. Then, Theorem \ref{thm: TheoremB} applies, and the linear perturbations are parity isospectral.
\end{lemmaB} 

\begin{proof}
Let $\{h^{(1)}_{nm},h^{(1)}_{mm}\}$ be a set of complex tetrad projections the metric perturbation that can be used to reconstruct a solution $(h^{(1)}_{\mu\nu},\varphi^{(1)}_{\mathcal{A}})$ of definite parity,
so that $\hat{P}h^{(1)}_{\mu\nu} = \pm \xi_\ell h^{(1)}_{\mu\nu}$ where
$\xi_{\ell}= (-1)^{\ell}$. 
We here follow the proof of Lemma~\ref{lem: ChiralLemma} by noting that the set 
$\{h^{(1)}_{nm},h^{(1)}_{mm}\}$ can be thought of as master variables (since they
can be used to reconstruct the metric perturbation).
Using $\hat{P}m^{\mu} = e^{i\chi}\bar{m}^{\mu}$, this relation yields
\begin{equation}\label{eq: DefParityComponents}
e^{i\chi}\,\hat{P}\big[h^{(1)}_{n\bar{m}}\big] = \pm \xi_{\ell}h^{(1)}_{nm}\,,
\qquad
e^{2i\chi}\,\hat{P}\big[h^{(1)}_{\bar{m}\bar{m}}\big] = \pm \xi_{\ell}h^{(1)}_{mm}\,.
\end{equation}
Now let $\tilde{h}^{(1)}_{nm} = ih^{(1)}_{nm}$ and
$\tilde{h}^{(1)}_{mm} = ih^{(1)}_{mm}$ be the
phase-rotated master variables. Since $\tilde{h}^{(1)}_{\mu\nu}$ and
$h^{(1)}_{\mu\nu}$ are real-valued, the conjugate components satisfy
$\tilde{h}^{(1)}_{n\bar{m}} = -i\,h^{(1)}_{n\bar{m}}$ and
$\tilde{h}^{(1)}_{\bar{m}\bar{m}} = -i\,h^{(1)}_{\bar{m}\bar{m}}$. Applying the
parity operator and using Eq.~\eqref{eq: DefParityComponents} we obtain,
\begin{subequations}
\begin{align}
e^{i\chi}\hat{P}\big[\tilde{h}^{(1)}_{n\bar{m}}\big]
&= -ie^{i\chi}\hat{P}\big[h^{(1)}_{n\bar{m}}\big]
= \mp i\xi_{\ell}h^{(1)}_{nm}
= \mp \xi_{\ell}\tilde{h}^{(1)}_{nm}\,,\\
e^{2i\chi}\hat{P}\big[\tilde{h}^{(1)}_{\bar{m}\bar{m}}\big]
&= -ie^{2i\chi} \hat{P}\big[h^{(1)}_{\bar{m}\bar{m}}\big]
= \mp i\xi_{\ell}h^{(1)}_{mm}
= \mp \xi_{\ell}\tilde{h}^{(1)}_{mm}\,.
\end{align}
\end{subequations}
Thus, if the complex set $\{h_{nm},h_{mm}\}$ correspond to definite parity solutions, then $\{i h_{nm},ih_{mm}\}$ will reconstruct the opposite parity solution and Theorem \ref{thm: TheoremB} is satisfied, establishing parity isospectrality. $\square$
\end{proof}

\vspace{0.2cm}
\noindent\textit{Isospectrality in vacuum GR}---We now verify that vacuum
perturbations of a Kerr black hole constitute a chiral-aligned theory and the QNMs are defined by complex-linear boundary conditions
satisfying the hypotheses of Lemma~\ref{lem: ChiralLemmaParity}, recovering the proof of
parity isospectrality at arbitrary subextremal spin~\cite{Li:2023ulk}.

Condition (i) of Def.~\ref{def: ChiralAligned} follows from a hidden
symmetry of Kerr: the rank-2 Killing tensor
$K_{\mu\nu}$~\cite{Carter:1968rr, Carter:1968ks}, whose conserved quantity
yields the Carter constant and integrable geodesic motion, factorizes into an
antisymmetric Killing--Yano tensor $f_{\alpha\beta}$, with
$K_{\mu\nu}=f_{\mu\alpha}f_\nu{}^{\alpha}$ and
$\nabla_{(\alpha}f_{\beta)\gamma}=0$~\cite{Penrose:1973um, Floyd:1973}.
Writing $f_{\mu\nu}\leftrightarrow f_{AA'BB'}
= \chi_{AB}\,\epsilon_{A'B'} + \bar{\chi}_{A'B'}\,\epsilon_{AB}$, the
self-dual part $\chi_{AB}$ is a non-null, valence-two Killing spinor
satisfying Eq.~\eqref{eq: CKCondition}, so the Kerr spacetime is conformally
K\"ahler~\cite{Aksteiner:2022bwr, Araneda:2025lak}. Conditions (ii)--(v)
hold automatically: the background is Ricci flat, and vacuum perturbations impose
$\Phi^{(1)}_{ij}=0$. Lemma~\ref{lem: ChiralLemma} therefore applies, with
master variables $\{h^{(1)}_{nm},h^{(1)}_{mm}\}$ in the ingoing radiation
gauge, and the linear perturbations are isospectral.

Furthermore, the Kinnersley tetrad \cite{Kinnersley:1969zza} is aligned with the repeated principal
null directions, realizing the background relations of
Eq.~\eqref{eq: CKFrameGauge}, and one verifies directly that it transforms
under parity as $\hat{P}l^{\mu}=l^{\mu}$, $\hat{P}n^{\mu}=n^{\mu}$, and
$\hat{P}m^{\mu}=-\bar{m}^{\mu}$, corresponding to $\chi=\pi$ in
Lemma~\ref{lem: ChiralLemmaParity}. 

Lastly, we show the boundary conditions defining the QNMs are complex-linear. The set of equations in Eq.~\eqref{eq: RicciFormCurved} can be combined to form the homogeneous Teukolsky equation $\mathcal{O}_0\Psi^{(1)}_0 = 0$ where $\Psi_0 = \Psi_{ABCD}o^{A}o^{B}o^{C}o^{D}$ is the $s = +2$ components of the self-dual Weyl spinor $\Psi_{ABCD}$ and $\mathcal{O}_0$ is a second-order Teukolsky operator [see Eq.~\eqref{eq: TekOperators} in the Supplemental material \cite{SupplementalMaterial}]. 

Separating $\Psi_{0}^{(1)}$ into radial and angular parts,
${}_{2}R_{\ell m}(r)$ and ${}_{2}S_{\ell m}(\theta)$, with
${}_{2}R_{\ell m}(r)$ purely ingoing at the horizon and purely outgoing at
null infinity and ${}_{2}S_{\ell m}(\theta)$ regular, defines a
Sturm--Liouville-type eigenvalue problem for the QNM frequencies. These
conditions are not posed independently on the real and imaginary parts of the
master variable, nor do they couple it to its complex conjugate, so the
eigenvalue problem respects the complex structure, as required by
Lemma~\ref{lem: ChiralLemmaParity}.

We emphasize the Bianchi identities formed in the outgoing radiation gauge with $\{h_{l\bar{m}},h_{\bar{m}\bar{m}}\}$ is also equivalent to the Teukolsky equation $\mathcal{O}_4\Psi^{(1)}_4 = 0$ where $\Psi_4 = \Psi_{ABCD}\iota^{A}\iota^{B}\iota^{C}\iota^{D}$ is the $s = -2$ component of the self-dual Weyl tensor and $\mathcal{O}_4$ is a second-order Teukolsky operator [see Supplemental Material \cite{SupplementalMaterial}]. Similarly, the boundary conditions are complex-linear. Parity isospectrality of Kerr then follows by applying Lemma \ref{lem: ChiralLemmaParity}, recovering the result of
Ref.~\cite{Li:2023ulk}.

Perturbations of Kerr realize the general construction of
Lemma~\ref{lem: ChiralLemmaParity} in a special way: rather than the coupled,
complex-linear system for the metric components
$\{h^{(1)}_{nm},h^{(1)}_{mm}\}$ of Eq.~\eqref{eq: RicciFormCurved}, the
equations of motion decouple into a single wave equation for the curvature
scalar $\Psi^{(1)}_{0}$ (or $\Psi^{(1)}_{4}$), from which the metric
components are recovered by a separate reconstruction procedure. References~\cite{Chrzanowski:1975wv, Chrzanowski:1976jy, Cohen_Kegeles_1975, Kegeles:1979an, Wald:1978vm} showed that the radiative parts of the linear metric perturbation of Kerr spacetime can be reconstructed from the Weyl scalars $\Psi^{(1)}_0$ or $\Psi^{(1)}_4$ in radiation gauges via an intermediate Hertz potential. Other metric reconstruction approaches include sequentially solving a subset of the Newman-Penrose equations \cite{Chandrasekhar_1983, Loutrel:2020wbw, Ripley:2020xby, Suvorov:2019qow, Li:2026rkf}. 

Regardless of the metric reconstruction procedure, the Weyl scalar perturbation $\Psi^{(1)}_0$ (or $\Psi^{(1)}_4$) can reconstruct the entire radiative metric perturbation, including both parity sectors of the gravitational wave. From this perspective, parity isospectrality is an immediate consequence: $\Psi^{(1)}_0$ (or $\Psi^{(1)}_4$) satisfies a single complex-linear wave equation, and both parity sectors are reconstructed from the same $\Psi^{(1)}_0$ (or $\Psi^{(1)}_4$), so the QNM spectrum of each parity must coincide if they are given the same boundary conditions.
 
\vspace{0.2cm}

\noindent 

{\it How to achieve isospectrality beyond GR}---We have formulated generic conditions that guarantee isospectrality in Theorem \ref{thm: MainTheorem} and for chiral-aligned theories in Lemma \ref{lem: ChiralLemma}. In this section we detail some features of linearized dynamics that ensure the conditions of a chiral-aligned theory, given in Def.~\ref{def: ChiralAligned}, are satisfied. 

Conditions (i) and
(ii) of Def.~\ref{def: ChiralAligned} constrain only the background
geometry; whether a modified theory preserves isospectrality therefore
hinges on the dynamical conditions (iii) and (iv). We now identify some sufficient conditions of a theory that enforce them: a component
condition on the modified linearized operator, which secures (iv),
and constraints on the effective stress-energy tensor, which secure
(iii).

One way to satisfy condition (iv) in Def.~\ref{def: ChiralAligned} for Ricci-flat, Petrov type D backgrounds in pure-metric, beyond-GR theories is for the modified linearized operator to act on the same set of metric components as the linearized Einstein operator. For example, let $\mathcal{E}_{\mu\nu}^{\rho\sigma}$ be the linearized Einstein tensor in GR \cite{Brito:2013yxa, Wardell:2024yoi} and
\begin{align} \label{eq: stress_Einstein}
    \mathcal{E}_{\mu \nu}^{\rho \sigma} h^{(1)}_{\rho \sigma} + \mathcal{T}_{\mu \nu}^{\rho \sigma} h^{(1)}_{\rho \sigma} = 0
\end{align}
be the modified linearized equations. We assume condition (iii) of Lemma \ref{def: ChiralAligned} holds, which allows the metric perturbation to be expressed in the ingoing or outgoing radiation gauge. The key step is to additionally require that the operator $\mathcal{T}_{\mu \nu}^{\rho \sigma}$ be nonvanishing only for the same components for which $\mathcal{E}_{\mu \nu}^{\rho \sigma}$ is nonvanishing.  This permits modifications of the equations of motion, while preserving the property that, in the ingoing radiation gauge, each perturbed Newman--Penrose Ricci scalar in Eq.~\eqref{eq: RicciFormCurved} depends only on $\{h^{(1)}_{nm},h^{(1)}_{mm}\}$, satisfying condition (iv) of Lemma \ref{lem: ChiralLemma}. Since this is a pure-metric modification,  condition (v) is also satisfied.

Condition (iii) of Def.~\ref{def: ChiralAligned} can be satisfied by constraining the equations of motion in the matter sector. Assuming a minimal coupling to the curvature, if $T^{\mathrm{eff}}_{ll}$ and $T^{\mathrm{eff}}_{nn}$ are quadratic in fields that vanish on the background, $\varphi^{(0)}_A=0$, then at $\mathcal{O}(\epsilon^1)$ these stress-energy components vanish. Thus, these scalars vanish as a consequence of their quadratic form without imposing or solving the matter equations of motion. For example, this occurs for the gravitoelectromagnetic perturbations of the Kerr--Newman black hole. Alternatively, if the theory obeys the null energy condition $T^{\mathrm{eff}}_{\mu\nu}k^\mu k^\nu \geq 0$, where $k^{\mu}$ is a null four-vector, and the background satisfies Eq.~\eqref{eq: CKRequire}, then on-shell $T^{\mathrm{eff}(1)}_{ll} = T^{\mathrm{eff}(1)}_{nn} = 0$ is enforced. 
Both mechanisms thus guarantee the extreme-boost-weight scalars vanish without spoiling the complex-linear structure.

\vspace{0.2cm}

\noindent\textit{Examples of the preservation of isospectrality}---As a demonstration of Theorem~\ref{thm: MainTheorem}, we discuss several theories and spacetimes in which it applies. We begin with the case where the criteria in both Theorem~\ref{thm: TheoremB} and Lemma~\ref{lem: ChiralLemmaParity} are satisfied. 

For perturbations of the Kerr--de Sitter black hole within GR, we can apply a similar argument to the Kerr case given above. The Kerr--de Sitter spacetime is both conformally K\"ahler and matter-aligned, permitting the gauge choice given in Eq.~\eqref{eq: CKFrameGauge}. As we show in the Supplemental Material, the equations of motion enforce that $\Phi^{(1)}_{ij}$ vanishes, allowing the metric perturbation to be transformable to the radiation gauges and the standard Teukolsky equation in \cite{Teukolsky:1973ha} to be satisfied. Together, these features imply that Lemma \ref{lem: ChiralLemmaParity} applies and linear perturbations must be parity isospectral. While a reconstruction procedure is guaranteed from Lemma \ref{lem: ChiralLemmaParity}, other techniques from \cite{Suvorov:2019qow, Loutrel:2020wbw, Ripley:2020xby, Li:2026rkf} can also be applied. 

Another example is the Kerr--Newman black hole within ordinary Einstein--Maxwell gravity. This spacetime is known to be conformally K\"ahler and to satisfy the matter-aligned condition. The linearized equations of motion can be put in the following form \cite{Dias:2015wqa}
\begin{subequations}\label{eq: KNEqus}
\begin{align}
\begin{aligned}
 \left(\mathcal{O}_{-2}+\Phi_{11}^{(0)} \mathcal{P}_{-2}\right) \Psi^{(1)}_4+\Phi_{11}^{(0)} \mathcal{Q}_{-2} \varphi^{(1)}_{-1}&=0\,, \\
 \left(\mathcal{O}_{-1}+\Phi_{11}^{(0)} \mathcal{P}_{-1}\right) \varphi^{(1)}_{-1}+\Phi_{11}^{(0)} \mathcal{Q}_{-1} \Psi^{(1)}_4&=0\,,
\end{aligned}
\end{align}
\end{subequations} 
where $\varphi_{-1}$ is the gauge-invariant, spin-1, master function and $\Phi^{(0)}_{11}$ is the background, Coulomb-type, Ricci scalar. The operators in Eq.~\eqref{eq: KNEqus} are presented in \cite{Dias:2015wqa} as well as the Supplemental Material in Eq.~\eqref{eq:  KNModifiedOperators}. 
The equations of motion in Eq.~\eqref{eq: KNEqus} are complex-linear because they couple to neither $\bar{\Psi}^{(1)}_4$ nor $\bar{\varphi}^{(1)}_{-1}$.    
In separate work \cite{Weller:2026kk}, we prove that metric perturbations of the Kerr--Newman black hole admit a traceless radiation gauge and provide a reconstruction procedure for both the metric and spin-1 gauge-field perturbations. These results establish that Lemma \ref{lem: ChiralLemmaParity} can be applied. While first steps were taken in~\cite{Pani:2013ija}, this structure of the linearized equations together with the reconstruction procedure shows that generic electrovacuum perturbations of the subextremal Kerr--Newman black hole are parity isospectral. 

Now let us consider an example that satisfies the criteria in Theorem~\ref{thm: TheoremB}, when restricted to the tensor sector, but not Lemma~\ref{lem: ChiralLemmaParity}, the metric formulation of the Starobinsky-Podolsky-higher-order action \cite{Cuzinatto:2016ehv}:
\begin{align}
     S[g] & = \int d^4 x \sqrt{|g|} \left(f(R) + c_1 \nabla_{\mu}R \nabla^{\mu}R \right.\\\
     &\left.+ \dots + c_N\nabla_{\mu_1}...\nabla_{\mu_N}R \nabla^{\mu_1}...\nabla^{\mu_N}R\right)  \nonumber
\end{align}
where $c_n$ are real coefficients, while $N$ is a positive integer. The Kerr black hole remains a stationary solution since the scalar curvature vanishes. As shown in the Supplemental Material, although the full metric perturbation of Kerr in this theory is not a vacuum solution and contains an additional massive Ricci mode, its radiative Weyl scalars obey the same homogeneous Teukolsky equations as vacuum perturbations of Kerr in GR \cite{Teukolsky:1973ha}. Thus, despite the theory not being chiral-aligned, the two tensor polarizations remain isospectral and share the standard vacuum Kerr spectrum.

Lastly, we emphasize an example where isospectrality has been demonstrated, but Theorem \ref{thm: MainTheorem} does not seem to apply. Recently, it was analytically shown, via a Chandrasekhar transformation, that perturbations of the Gibbons and Maeda black hole in Einstein-Maxwell-dilaton-axion gravity are isospectral \cite{Pope:2025jgz}. As previously mentioned, this black hole is conformally K\"ahler and does not satisfy the matter-aligned condition. Additionally, the theory does not satisfy any of the other criteria of a chiral-aligned theory. Moreover, we cannot identify a complex-linear formulation of the perturbations. In spite of all of this, isospectrality remains. Thus, this example appears to realize a different mechanism for isospectrality, which requires the existence of an intertwining relation. While the conditions of Theorem \ref{thm: MainTheorem} do not apply within the usual Newman--Penrose formalism, it is possible a different combination of complex variables, possibly in a different framework, in metric components or the curvature, could apply. We leave this question to future work. 

\vspace{0.2cm}
\noindent\textit{Conclusion}---In this letter, we have shown that isospectrality follows whenever the
linearized equations of motion reduce to a closed, complex-linear system
whose solutions reconstruct the radiative fields, with boundary conditions
that preserve the complex structure in Theorem \ref{thm: MainTheorem}, and that parity
isospectrality follows when the associated phase rotation exchanges the
parity sectors in Theorem \ref{thm: TheoremB}. The chiral-aligned class, given in Def.~\ref{def: ChiralAligned}, supplies a broad set of theories satisfying the criterion of Theorem \ref{thm: MainTheorem} and, with admissible backgrounds, Theorem \ref{thm: TheoremB}. Within it, our framework recovers the known parity isospectrality of Kerr, and additionally that of Kerr--de Sitter, and extends to Kerr--Newman spacetime in Einstein-Maxwell theory~\cite{Newman:1965my,Dias:2015wqa}. Starting from the structure of the equations of motion directly, Theorem \ref{thm: MainTheorem} establishes a general sufficient criterion for isospectrality given any action in the form of Eq.~\eqref{eq: Action}, without relying on the slow-rotation expansion or on constructing an explicit transformation between parity sectors. Our result also hints at why isospectrality is a fragile symmetry in many theories or spacetimes beyond vacuum GR \cite{Cano:2023jbk, Cano:2024ezp, Li:2023ulk, Li:2025fci, Chung:2024ira, Cardoso:2002pa, Kodama:2003jz, Cardoso:2001bb, Berti:2003ud, Weller:2024qvo, Wu:2025obg, Saketh:2024ojw, Pani:2009ss, Laeuger:2025zgb, Thorne:1968zz, Kokkotas:1999bd}, where the complex-linear structure of the GR perturbation equations is typically spoiled by algebraically-general background spacetimes.

Recent work has proposed that the guiding principle that ensures isospectrality should be the decoupling of self-dual variables \cite{Pereniguez:2026avs}. This is consistent with, and a limiting case of our findings as long as the equations do not couple to the anti-self-dual sector, which for real field configurations is the complex conjugate. Our criteria, however, are more general, since it requires only a complex-linear master system that reconstructs the radiative fields, whether or not that system arises from self-dual variables. Other recent work proposes a mechanism for breaking isospectrality in the eikonal limit of large multipole number $\ell$ \cite{Cano:2024wzo, Bah:2026aia}, conjecturing that isospectrality is equivalent to the absence of birefringence. While birefringence generically indicates the breaking of isospectrality, our work suggests the converse is false. For example, a Schwarzschild black hole with a thin shell breaks isospectrality \cite{Laeuger:2025zgb}, yet exhibits no amplitude, phase, or velocity birefringence. Reference~\cite{Bah:2026aia} reaches a similar conclusion, interpreting the absence of birefringence as an electromagnetic-type duality at the light ring. These studies, however, capture only leading-order $1/\ell$ contributions to the QNM spectrum and do not characterize isospectrality for all radiative multipoles. Our results are more general, guaranteeing isospectrality for every radiative multipole in the class of configurations satisfying Theorem~\ref{thm: MainTheorem}.

There are several avenues for future work. A clear direction is to search for other classes of theories and black holes that also yield complex-linear equations of motion, and other reconstruction methods can be implemented. Complex substructures in self-dual gravity \cite{Plebanski:1975wn, Green:2026nlt} and twistor theory \cite{Penrose:1967wn} may be helpful as well as hidden symmetries of the classical double-copy \cite{Monteiro:2014cda, Kent:2025pvu}. 

\vspace{0.4cm}
\noindent\textit{Acknowledgments---}
We thank Hector Silva for interesting discussions. C.~W., A.~L., and Y.~C.'s research is supported by the Simons Foundation (Award No. 568762),
the Brinson Foundation, and the National Science Foun-
dation (via Grants No. PHY-2011961 and No. PHY-
2011968). A.L. acknowledges support from the Fannie and John Hertz Foundation in the form of a Hertz Fellowship. D.~L. and N.~Y. acknowledge support from the Simons Foundation (via Award No. 896696), the Simons Foundation International (via Grant No. SFI-MPS-BH-00012593-01), and the NSF (via Grants No.~PHY-2512423). P.~W. acknowledges funding from the Deutsche Forschungsgemeinschaft (DFG) project number: 386119226.

\bibliographystyle{apsrev4-1}
\bibliography{reference}

@article{Carter:1968rr,
    author = "Carter, Brandon",
    title = "{Global structure of the Kerr family of gravitational fields}",
    doi = "10.1103/PhysRev.174.1559",
    journal = "Phys. Rev.",
    volume = "174",
    pages = "1559--1571",
    year = "1968"
}

@article{Carter:1968ks,
    author = "Carter, B.",
    title = "{Hamilton-Jacobi and Schrodinger separable solutions of Einstein's equations}",
    doi = "10.1007/BF03399503",
    journal = "Commun. Math. Phys.",
    volume = "10",
    number = "4",
    pages = "280--310",
    year = "1968"
}

@article{Cuzinatto:2016ehv,
    author = "Cuzinatto, R. R. and de Melo, C. A. M. and Medeiros, L. G. and Pompeia, P. J.",
    title = "{Scalar-multi-tensorial equivalence for higher order  $f\left( R,\nabla_{\mu} R,\nabla_{\mu_{1}}\nabla_{\mu_{2}}R,...,\nabla_{\mu_{1}}...\nabla_{\mu_{n} }R\right)$ theories of gravity}",
    eprint = "1603.01563",
    archivePrefix = "arXiv",
    primaryClass = "gr-qc",
    doi = "10.1103/PhysRevD.93.124034",
    journal = "Phys. Rev. D",
    volume = "93",
    number = "12",
    pages = "124034",
    year = "2016",
    note = "[Erratum: Phys.Rev.D 98, 029901 (2018)]"
}

@article{Suzuki:1998vy,
    author = "Suzuki, Hisao and Takasugi, Eiichi and Umetsu, Hiroshi",
    title = "{Perturbations of Kerr-de Sitter black hole and Heun's equations}",
    eprint = "gr-qc/9805064",
    archivePrefix = "arXiv",
    reportNumber = "EPHOU-98-005, OU-HET-296",
    doi = "10.1143/PTP.100.491",
    journal = "Prog. Theor. Phys.",
    volume = "100",
    pages = "491--505",
    year = "1998"
}

@article{Konoplya:2007zx,
    author = "Konoplya, R. A. and Zhidenko, A.",
    title = "{Decay of a charged scalar and Dirac fields in the Kerr-Newman-de Sitter background}",
    eprint = "0707.1890",
    archivePrefix = "arXiv",
    primaryClass = "hep-th",
    doi = "10.1103/PhysRevD.76.084018",
    journal = "Phys. Rev. D",
    volume = "76",
    number = "8",
    pages = "084018",
    year = "2007",
    note = "[Erratum: Phys.Rev.D 90, 029901 (2014)]"
}

@article{Pope:2025jgz,
    author = "Pope, C. N. and Rohrer, D. O. and Whiting, B. F.",
    title = "{Perturbations of black holes in Einstein-Maxwell-dilaton-axion theories}",
    eprint = "2508.04589",
    archivePrefix = "arXiv",
    primaryClass = "hep-th",
    reportNumber = "MI-HET-864",
    doi = "10.1103/1b6k-f38p",
    journal = "Phys. Rev. D",
    volume = "112",
    number = "12",
    pages = "124064",
    year = "2025"
}

@article{Thorne:1968zz,
    author = "Thorne, Kip S.",
    title = "{Gravitational Radiation Damping}",
    doi = "10.1103/PhysRevLett.21.320",
    journal = "Phys. Rev. Lett.",
    volume = "21",
    pages = "320--323",
    year = "1968"
}

@article{Saketh:2024ojw,
    author = "Saketh, M. V. S. and Maggio, Elisa",
    title = "{Quasinormal modes of slowly-spinning horizonless compact objects}",
    eprint = "2406.10070",
    archivePrefix = "arXiv",
    primaryClass = "gr-qc",
    doi = "10.1103/PhysRevD.110.084038",
    journal = "Phys. Rev. D",
    volume = "110",
    number = "8",
    pages = "084038",
    year = "2024"
}

@article{Davis:1971kk,
  title = {Gravitational Radiation from a Particle Falling Radially into a Schwarzschild Black Hole},
  author = {Davis, Marc and Ruffini, Remo and Press, William H. and Price, Richard H.},
  journal = {Phys. Rev. Lett.},
  volume = {27},
  issue = {21},
  pages = {1466--1469},
  numpages = {0},
  year = {1971},
  month = {Nov},
  publisher = {American Physical Society},
  doi = {10.1103/PhysRevLett.27.1466},
  url = {https://link.aps.org/doi/10.1103/PhysRevLett.27.1466}
}

@article{Wu:2025obg,
    author = "Wu, David G. and Hussain, Asad and Zimmerman, Aaron",
    title = "{Computing spectral shifts for Johannsen-Psaltis black holes}",
    eprint = "2512.14679",
    archivePrefix = "arXiv",
    primaryClass = "gr-qc",
    doi = "10.1103/3hm4-89xc",
    journal = "Phys. Rev. D",
    volume = "113",
    number = "10",
    pages = "104059",
    year = "2026"
}

@article{Berti:2025hly,
    author = "Berti, Emanuele and others",
    editor = "Berti, Emanuele and Cardoso, Vitor and Carullo, Gregorio",
    title = "{Black hole spectroscopy: from theory to experiment}",
    eprint = "2505.23895",
    archivePrefix = "arXiv",
    primaryClass = "gr-qc",
    reportNumber = "RUP-25-10; YITP-25-64; RIKEN-iTHEMS-Report-25",
    doi = "10.1088/1361-6382/ae59e2",
    journal = "Class. Quant. Grav.",
    volume = "43",
    number = "12",
    pages = "123001",
    year = "2026"
}

@misc{SupplementalMaterial,
  title        = {Supplemental Material for ``Isospectrality from Chirality''},
  note         = {URL insert},
}

@article{PLEBANSKI197698,
title = {Rotating, charged, and uniformly accelerating mass in general relativity},
journal = {Annals of Physics},
volume = {98},
number = {1},
pages = {98-127},
year = {1976},
issn = {0003-4916},
doi = {https://doi.org/10.1016/0003-4916(76)90240-2},
url = {https://www.sciencedirect.com/science/article/pii/0003491676902402},
author = {J.F Plebanski and M Demianski}
}

@article{Bah:2026aia,
    author = "Bah, Ibrahima and Berti, Emanuele and De Luca, Valerio and Ganchev, Bogdan and Pere{\~n}iguez, David",
    title = "{Gravitational electric-magnetic duality at the light ring and quasinormal mode isospectrality in effective field theories}",
    eprint = "2605.03018",
    archivePrefix = "arXiv",
    primaryClass = "gr-qc",
    month = "5",
    year = "2026"
}

@article{Kodama:2003jz,
    author = "Kodama, Hideo and Ishibashi, Akihiro",
    title = "{A Master equation for gravitational perturbations of maximally symmetric black holes in higher dimensions}",
    eprint = "hep-th/0305147",
    archivePrefix = "arXiv",
    doi = "10.1143/PTP.110.701",
    journal = "Prog. Theor. Phys.",
    volume = "110",
    pages = "701--722",
    year = "2003"
}

@article{Penrose:1967wn,
    author = "Penrose, R.",
    title = "{Twistor algebra}",
    doi = "10.1063/1.1705200",
    journal = "J. Math. Phys.",
    volume = "8",
    pages = "345",
    year = "1967"
}

@article{Plebanski:1975wn,
    author = "Plebanski, J. F.",
    title = "{Some solutions of complex Einstein equations}",
    doi = "10.1063/1.522505",
    journal = "J. Math. Phys.",
    volume = "16",
    pages = "2395--2402",
    year = "1975"
}

@article{Weller:2026kk,
  author       = {Weller, Colin and Li, Dongjun and Wagle, Pratik and Chen, Yanbei and Yunes, Nic\'olas},
  title        = {Metric Reconstruction for Generic Perturbations of the Kerr--Newman Black Hole},
  note         = {To be published},
  year         = {2026},
}

@article{Pani:2013ija,
    author = "Pani, Paolo and Berti, Emanuele and Gualtieri, Leonardo",
    title = "{Gravitoelectromagnetic Perturbations of Kerr-Newman Black Holes: Stability and Isospectrality in the Slow-Rotation Limit}",
    eprint = "1304.1160",
    archivePrefix = "arXiv",
    primaryClass = "gr-qc",
    doi = "10.1103/PhysRevLett.110.241103",
    journal = "Phys. Rev. Lett.",
    volume = "110",
    number = "24",
    pages = "241103",
    year = "2013"
}

@article{Mitman:2025hgy,
    author = "Mitman, Keefe and others",
    title = "{Probing the ringdown perturbation in binary black hole coalescences with an improved quasinormal mode extraction algorithm}",
    eprint = "2503.09678",
    archivePrefix = "arXiv",
    primaryClass = "gr-qc",
    doi = "10.1103/qq1g-jlnw",
    journal = "Phys. Rev. D",
    volume = "112",
    number = "6",
    pages = "064016",
    year = "2025"
}

@article{Li:2025fci,
    author = "Li, Dongjun and Wagle, Pratik and Chen, Yanbei and Yunes, Nicol{\'a}s",
    title = "{Perturbations of spinning black holes in dynamical Chern-Simons gravity: Slow rotation quasinormal modes}",
    eprint = "2503.15606",
    archivePrefix = "arXiv",
    primaryClass = "gr-qc",
    doi = "10.1103/yfw6-32yl",
    journal = "Phys. Rev. D",
    volume = "112",
    number = "4",
    pages = "044005",
    year = "2025"
}

@article{Wald:1973wwa,
    author = "Wald, Robert M.",
    title = "{On perturbations of a Kerr black hole}",
    doi = "10.1063/1.1666203",
    journal = "J. Math. Phys.",
    volume = "14",
    number = "10",
    pages = "1453--1461",
    year = "1973"
}

@article{Akcay:2010vt,
    author = "Akcay, Sarp and Matzner, Richard A.",
    title = "{Kerr-de Sitter Universe}",
    eprint = "1011.0479",
    archivePrefix = "arXiv",
    primaryClass = "gr-qc",
    doi = "10.1088/0264-9381/28/8/085012",
    journal = "Class. Quant. Grav.",
    volume = "28",
    pages = "085012",
    year = "2011"
}

@article{Frolov:2017kze,
    author = "Frolov, Valeri P. and Krtous, Pavel and Kubiznak, David",
    title = "{Black holes, hidden symmetries, and complete integrability}",
    eprint = "1705.05482",
    archivePrefix = "arXiv",
    primaryClass = "gr-qc",
    doi = "10.1007/s41114-017-0009-9",
    journal = "Living Rev. Rel.",
    volume = "20",
    number = "1",
    pages = "6",
    year = "2017"
}

@article{Dias:2013sdc,
    author = "Dias, {\'O}scar J. C. and Santos, Jorge E.",
    title = "{Boundary Conditions for Kerr-AdS Perturbations}",
    eprint = "1302.1580",
    archivePrefix = "arXiv",
    primaryClass = "hep-th",
    doi = "10.1007/JHEP10(2013)156",
    journal = "JHEP",
    volume = "10",
    pages = "156",
    year = "2013"
}

@article{Henneaux:1985tv,
    author = "Henneaux, M. and Teitelboim, C.",
    title = "{Asymptotically anti-De Sitter Spaces}",
    doi = "10.1007/BF01205790",
    journal = "Commun. Math. Phys.",
    volume = "98",
    pages = "391--424",
    year = "1985"
}

@article{Kinnersley:1969zza,
    author = "Kinnersley, William",
    title = "{Type D Vacuum Metrics}",
    doi = "10.1063/1.1664958",
    journal = "J. Math. Phys.",
    volume = "10",
    pages = "1195--1203",
    year = "1969"
}

@article{Darboux1882,
  author  = {Darboux, G.},
  title   = {Sur une proposition relative aux \'equations lin\'eaires},
  journal = {C. R. Acad. Sci. (Paris)},
  volume  = {94},
  pages   = {1456},
  year    = {1882}
}

@article{Golfand:1971iw,
    author = "Golfand, Yu. A. and Likhtman, E. P.",
    editor = "Salam, A. and Sezgin, E.",
    title = "{Extension of the Algebra of Poincare Group Generators and Violation of p Invariance}",
    doi = "10.1142/9789814542340_0001",
    journal = "JETP Lett.",
    volume = "13",
    pages = "323--326",
    year = "1971"
}

@inbook{FernandezC:2018cdo,
    author = "Fernandez C, David J.",
    title = "{Trends in Supersymmetric Quantum Mechanics}",
    eprint = "1811.06449",
    archivePrefix = "arXiv",
    primaryClass = "quant-ph",
    doi = "10.1007/978-3-030-20087-9_2",
    year = "2019"
}

@book{Cooper:2001weo,
    author = "Cooper, Fred and Khare, Avinash and Sukhatme, Uday",
    title = "{Supersymmetry and quantum mechanics}",
    doi = "10.1142/4687",
    publisher = "World Scientific",
    month = "6",
    year = "2001"
}

@article{Cooper:1982dm,
    author = "Cooper, Fred and Freedman, Barry",
    title = "{Aspects of Supersymmetric Quantum Mechanics}",
    reportNumber = "LA-UR-82-2144",
    doi = "10.1016/0003-4916(83)90034-9",
    journal = "Annals Phys.",
    volume = "146",
    pages = "262",
    year = "1983"
}

@article{Witten:1981nf,
    author = "Witten, Edward",
    title = "{Dynamical Breaking of Supersymmetry}",
    reportNumber = "Print-81-0317 (PRINCETON)",
    doi = "10.1016/0550-3213(81)90006-7",
    journal = "Nucl. Phys. B",
    volume = "188",
    pages = "513",
    year = "1981"
}

@article{Penrose:1960eq,
    author = "Penrose, R.",
    title = "{A Spinor approach to general relativity}",
    doi = "10.1016/0003-4916(60)90021-X",
    journal = "Annals Phys.",
    volume = "10",
    pages = "171--201",
    year = "1960"
}

@article{Dias:2015wqa,
    author = "Dias, Oscar J. C. and Godazgar, Mahdi and Santos, Jorge E.",
    title = "{Linear Mode Stability of the Kerr-Newman Black Hole and Its Quasinormal Modes}",
    eprint = "1501.04625",
    archivePrefix = "arXiv",
    primaryClass = "gr-qc",
    doi = "10.1103/PhysRevLett.114.151101",
    journal = "Phys. Rev. Lett.",
    volume = "114",
    number = "15",
    pages = "151101",
    year = "2015"
}

@article{Isi:2019aib,
    author = "Isi, Maximiliano and Giesler, Matthew and Farr, Will M. and Scheel, Mark A. and Teukolsky, Saul A.",
    title = "{Testing the no-hair theorem with GW150914}",
    eprint = "1905.00869",
    archivePrefix = "arXiv",
    primaryClass = "gr-qc",
    reportNumber = "LIGO-P1900135",
    doi = "10.1103/PhysRevLett.123.111102",
    journal = "Phys. Rev. Lett.",
    volume = "123",
    number = "11",
    pages = "111102",
    year = "2019"
}

@article{Price:2006ke,
    author = "Price, Larry R. and Shankar, Karthik and Whiting, Bernard F.",
    title = "{On the existence of radiation gauges in Petrov type II spacetimes}",
    eprint = "gr-qc/0611070",
    archivePrefix = "arXiv",
    doi = "10.1088/0264-9381/24/9/014",
    journal = "Class. Quant. Grav.",
    volume = "24",
    pages = "2367--2388",
    year = "2007"
}

@article{Li:2023ulk,
    author = "Li, Dongjun and Hussain, Asad and Wagle, Pratik and Chen, Yanbei and Yunes, Nicol\'as and Zimmerman, Aaron",
    title = "{Isospectrality breaking in the Teukolsky formalism}",
    eprint = "2310.06033",
    archivePrefix = "arXiv",
    primaryClass = "gr-qc",
    reportNumber = "UTWI-14-2023",
    doi = "10.1103/PhysRevD.109.104026",
    journal = "Phys. Rev. D",
    volume = "109",
    number = "10",
    pages = "104026",
    year = "2024"
}

@article{Wagle:2023fwl,
    author = "Wagle, Pratik and Li, Dongjun and Chen, Yanbei and Yunes, Nicolas",
    title = "{Perturbations of spinning black holes in dynamical Chern-Simons gravity: Slow rotation equations}",
    eprint = "2311.07706",
    archivePrefix = "arXiv",
    primaryClass = "gr-qc",
    doi = "10.1103/PhysRevD.109.104029",
    journal = "Phys. Rev. D",
    volume = "109",
    number = "10",
    pages = "104029",
    year = "2024"
}

@article{Cano:2023jbk,
    author = "Cano, Pablo A. and Fransen, Kwinten and Hertog, Thomas and Maenaut, Simon",
    title = "{Quasinormal modes of rotating black holes in higher-derivative gravity}",
    eprint = "2307.07431",
    archivePrefix = "arXiv",
    primaryClass = "gr-qc",
    doi = "10.1103/PhysRevD.108.124032",
    journal = "Phys. Rev. D",
    volume = "108",
    number = "12",
    pages = "124032",
    year = "2023"
}

@article{Maselli:2015tta,
    author = "Maselli, Andrea and Pani, Paolo and Gualtieri, Leonardo and Ferrari, Valeria",
    title = "{Rotating black holes in Einstein-Dilaton-Gauss-Bonnet gravity with finite coupling}",
    eprint = "1507.00680",
    archivePrefix = "arXiv",
    primaryClass = "gr-qc",
    doi = "10.1103/PhysRevD.92.083014",
    journal = "Phys. Rev. D",
    volume = "92",
    number = "8",
    pages = "083014",
    year = "2015"
}

@article{Chandrasekhar:1975nkd,
    author = "Chandrasekhar, Subrahmanyan",
    title = "{On the equations governing the perturbations of the Schwarzschild black hole}",
    doi = "10.1098/rspa.1975.0066",
    journal = "Proc. Roy. Soc. Lond. A",
    volume = "343",
    number = "1634",
    pages = "289--298",
    year = "1975"
}

@article{Cardoso:2001bb,
    author = "Cardoso, Vitor and Lemos, Jose P. S.",
    title = "{Quasinormal modes of Schwarzschild anti-de Sitter black holes: Electromagnetic and gravitational perturbations}",
    eprint = "gr-qc/0105103",
    archivePrefix = "arXiv",
    reportNumber = "DF-IST-4-2001",
    doi = "10.1103/PhysRevD.64.084017",
    journal = "Phys. Rev. D",
    volume = "64",
    pages = "084017",
    year = "2001"
}

@book{Penrose:1986ca,
    author = "Penrose, R. and Rindler, W.",
    title = "{SPINORS AND SPACE-TIME. VOL. 2: SPINOR AND TWISTOR METHODS IN SPACE-TIME GEOMETRY}",
    doi = "10.1017/CBO9780511524486",
    isbn = "978-0-521-34786-0, 978-0-511-86842-9",
    publisher = "Cambridge University Press",
    series = "Cambridge Monographs on Mathematical Physics",
    month = "4",
    year = "1988"
}

@ARTICLE{Goldberg2009,
       author = {{Goldberg}, J.~N. and {Sachs}, R.~K.},
        title = "{Republication of: A theorem on Petrov types}",
      journal = {General Relativity and Gravitation},
         year = 2009,
        month = feb,
       volume = {41},
       number = {2},
        pages = {433-444},
          doi = {10.1007/s10714-008-0722-5},
       adsurl = {https://ui.adsabs.harvard.edu/abs/2009GReGr..41..433G}
}

@article{Franchini:2025csk,
    author = "Franchini, Nicola",
    title = "{Quasinormal Modes Ratios as Agnostic Test of General Relativity}",
    eprint = "2511.02010",
    archivePrefix = "arXiv",
    primaryClass = "gr-qc",
    doi = "10.1103/cm6z-43t3",
    journal = "Phys. Rev. Lett.",
    volume = "136",
    number = "22",
    pages = "221401",
    year = "2026"
}

@article{Brito:2013yxa,
    author = "Brito, Richard and Cardoso, Vitor and Pani, Paolo",
    title = "{Partially massless gravitons do not destroy general relativity black holes}",
    eprint = "1306.0908",
    archivePrefix = "arXiv",
    primaryClass = "gr-qc",
    doi = "10.1103/PhysRevD.87.124024",
    journal = "Phys. Rev. D",
    volume = "87",
    number = "12",
    pages = "124024",
    year = "2013"
}

@article{Ripley:2020xby,
    author = "Ripley, Justin L. and Loutrel, Nicholas and Giorgi, Elena and Pretorius, Frans",
    title = "{Numerical computation of second order vacuum perturbations of Kerr black holes}",
    eprint = "2010.00162",
    archivePrefix = "arXiv",
    primaryClass = "gr-qc",
    doi = "10.1103/PhysRevD.103.104018",
    journal = "Phys. Rev. D",
    volume = "103",
    pages = "104018",
    year = "2021"
}

@article{Wald:1978vm,
    author = "Wald, Robert M.",
    title = "{Construction of Solutions of Gravitational, Electromagnetic, Or Other Perturbation Equations from Solutions of Decoupled Equations}",
    doi = "10.1103/PhysRevLett.41.203",
    journal = "Phys. Rev. Lett.",
    volume = "41",
    pages = "203--206",
    year = "1978"
}

@article{Chrzanowski:1976jy,
    author = "Chrzanowski, P. L.",
    title = "{Applications of Metric Perturbations of a Rotating Black Hole: Distortion of the Event Horizon}",
    doi = "10.1103/PhysRevD.13.806",
    journal = "Phys. Rev. D",
    volume = "13",
    pages = "806--818",
    year = "1976"
}

@article{Li:2022pcy,
    author = "Li, Dongjun and Wagle, Pratik and Chen, Yanbei and Yunes, Nicol\'as",
    title = "{Perturbations of Spinning Black Holes beyond General Relativity: Modified Teukolsky Equation}",
    eprint = "2206.10652",
    archivePrefix = "arXiv",
    primaryClass = "gr-qc",
    doi = "10.1103/PhysRevX.13.021029",
    journal = "Phys. Rev. X",
    volume = "13",
    number = "2",
    pages = "021029",
    year = "2023"
}

@article{Cohen_Kegeles_1975, 
    title={Space-time perturbations}, 
    volume={54}, 
    ISSN={0375-9601}, 
    DOI={10.1016/0375-9601(75)90583-6}, 
    number={1}, 
    journal={Physics Letters A}, 
    author={Cohen, J. M. and Kegeles, L. S.}, 
    year={1975}, month={Aug}, 
    pages={5-7}
}

@article{Nichols:2012jn,
    author = "Nichols, David A. and Zimmerman, Aaron and Chen, Yanbei and Lovelace, Geoffrey and Matthews, Keith D. and Owen, Robert and Zhang, Fan and Thorne, Kip S.",
    title = "{Visualizing Spacetime Curvature via Frame-Drag Vortexes and Tidal Tendexes III. Quasinormal Pulsations of Schwarzschild and Kerr Black Holes}",
    eprint = "1208.3038",
    archivePrefix = "arXiv",
    primaryClass = "gr-qc",
    doi = "10.1103/PhysRevD.86.104028",
    journal = "Phys. Rev. D",
    volume = "86",
    pages = "104028",
    year = "2012"
}

@article{Stein:2014xba,
    author = "Stein, Leo C.",
    title = "{Rapidly rotating black holes in dynamical Chern-Simons gravity: Decoupling limit solutions and breakdown}",
    eprint = "1407.2350",
    archivePrefix = "arXiv",
    primaryClass = "gr-qc",
    doi = "10.1103/PhysRevD.90.044061",
    journal = "Phys. Rev. D",
    volume = "90",
    number = "4",
    pages = "044061",
    year = "2014"
}

@article{Dolan:2009nk,
    author = "Dolan, Sam R. and Ottewill, Adrian C.",
    editor = "Uranga, A. M.",
    title = "{On an Expansion Method for Black Hole Quasinormal Modes and Regge Poles}",
    eprint = "0908.0329",
    archivePrefix = "arXiv",
    primaryClass = "gr-qc",
    doi = "10.1088/0264-9381/26/22/225003",
    journal = "Class. Quant. Grav.",
    volume = "26",
    pages = "225003",
    year = "2009"
}

@article{Petrov:2000bs,
    author = "Petrov, A. Z.",
    title = "{The Classification of spaces defining gravitational fields}",
    doi = "10.1023/A:1001910908054",
    journal = "Gen. Rel. Grav.",
    volume = "32",
    pages = "1661--1663",
    year = "2000"
}

@article{Kegeles:1979an,
    author = "Kegeles, L. S. and Cohen, J. M.",
    title = "{CONSTRUCTIVE PROCEDURE FOR PERTURBATIONS OF SPACE-TIMES}",
    doi = "10.1103/PhysRevD.19.1641",
    journal = "Phys. Rev. D",
    volume = "19",
    pages = "1641--1664",
    year = "1979"
}

@article{Cardoso:2018ptl,
    author = "Cardoso, Vitor and Kimura, Masashi and Maselli, Andrea and Senatore, Leonardo",
    title = "{Black Holes in an Effective Field Theory Extension of General Relativity}",
    eprint = "1808.08962",
    archivePrefix = "arXiv",
    primaryClass = "gr-qc",
    doi = "10.1103/PhysRevLett.121.251105",
    journal = "Phys. Rev. Lett.",
    volume = "121",
    number = "25",
    pages = "251105",
    year = "2018"
}

@article{Tattersall:2018nve,
    author = "Tattersall, Oliver J. and Ferreira, Pedro G.",
    title = "{Quasinormal modes of black holes in Horndeski gravity}",
    eprint = "1804.08950",
    archivePrefix = "arXiv",
    primaryClass = "gr-qc",
    doi = "10.1103/PhysRevD.97.104047",
    journal = "Phys. Rev. D",
    volume = "97",
    number = "10",
    pages = "104047",
    year = "2018"
}

@article{Newman:1961qr,
    author = "Newman, Ezra and Penrose, Roger",
    title = "{An Approach to gravitational radiation by a method of spin coefficients}",
    doi = "10.1063/1.1724257",
    journal = "J. Math. Phys.",
    volume = "3",
    pages = "566--578",
    year = "1962"
}

@article{Franchini:2022axs,
    author = {Franchini, Nicola and V\"olkel, Sebastian H.},
    title = "{Parametrized quasinormal mode framework for non-Schwarzschild metrics}",
    eprint = "2210.14020",
    archivePrefix = "arXiv",
    primaryClass = "gr-qc",
    doi = "10.1103/PhysRevD.107.124063",
    journal = "Phys. Rev. D",
    volume = "107",
    number = "12",
    pages = "124063",
    year = "2023"
}

@book{Thorne:1986iy,
    editor = "Thorne, Kip S. and Price, R. H. and Macdonald, D. A.",
    title = "{BLACK HOLES: THE MEMBRANE PARADIGM}",
    isbn = "978-0-300-03770-8",
    year = "1986"
}

@article{Price:1971fb,
    author = "Price, Richard H.",
    title = "{Nonspherical perturbations of relativistic gravitational collapse. 1. Scalar and gravitational perturbations}",
    doi = "10.1103/PhysRevD.5.2419",
    journal = "Phys. Rev. D",
    volume = "5",
    pages = "2419--2438",
    year = "1972"
}

@article{Robinson:1975bv,
    author = "Robinson, D. C.",
    title = "{Uniqueness of the Kerr black hole}",
    doi = "10.1103/PhysRevLett.34.905",
    journal = "Phys. Rev. Lett.",
    volume = "34",
    pages = "905--906",
    year = "1975"
}

@article{Carter:1971zc,
    author = "Carter, B.",
    title = "{Axisymmetric Black Hole Has Only Two Degrees of Freedom}",
    doi = "10.1103/PhysRevLett.26.331",
    journal = "Phys. Rev. Lett.",
    volume = "26",
    pages = "331--333",
    year = "1971"
}

@article{Israel:1967wq,
    author = "Israel, Werner",
    title = "{Event horizons in static vacuum space-times}",
    doi = "10.1103/PhysRev.164.1776",
    journal = "Phys. Rev.",
    volume = "164",
    pages = "1776--1779",
    year = "1967"
}

@article{Teukolsky:1974yv,
    author = "Teukolsky, S. A. and Press, W. H.",
    title = "{Perturbations of a rotating black hole. III - Interaction of the hole with gravitational and electromagnetic radiation}",
    doi = "10.1086/153180",
    journal = "Astrophys. J.",
    volume = "193",
    pages = "443--461",
    year = "1974"
}

@article{Cano:2020cao,
    author = "Cano, Pablo A. and Fransen, Kwinten and Hertog, Thomas",
    title = "{Ringing of rotating black holes in higher-derivative gravity}",
    eprint = "2005.03671",
    archivePrefix = "arXiv",
    primaryClass = "gr-qc",
    doi = "10.1103/PhysRevD.102.044047",
    journal = "Phys. Rev. D",
    volume = "102",
    number = "4",
    pages = "044047",
    year = "2020"
}

@article{Kokkotas:1999bd,
    author = "Kokkotas, Kostas D. and Schmidt, Bernd G.",
    title = "{Quasinormal modes of stars and black holes}",
    eprint = "gr-qc/9909058",
    archivePrefix = "arXiv",
    doi = "10.12942/lrr-1999-2",
    journal = "Living Rev. Rel.",
    volume = "2",
    pages = "2",
    year = "1999"
}

@article{Leaver:1985ax,
    author = "Leaver, E. W.",
    title = "{An analytic representation for the quasi normal modes of Kerr black holes}",
    doi = "10.1098/rspa.1985.0119",
    journal = "Proc. Roy. Soc. Lond. A",
    volume = "402",
    pages = "285--298",
    year = "1985"
}

@article{Aksteiner:2022bwr,
    author = "Aksteiner, Steffen and Araneda, Bernardo",
    title = {{K{\"a}hler Geometry of Black Holes and Gravitational Instantons}},
    eprint = "2207.10039",
    archivePrefix = "arXiv",
    primaryClass = "gr-qc",
    doi = "10.1103/PhysRevLett.130.161502",
    journal = "Phys. Rev. Lett.",
    volume = "130",
    number = "16",
    pages = "161502",
    year = "2023"
}

@article{Newman:1965my,
    author = "Newman, E T. and Couch, R. and Chinnapared, K. and Exton, A. and Prakash, A. and Torrence, R.",
    title = "{Metric of a Rotating, Charged Mass}",
    doi = "10.1063/1.1704351",
    journal = "J. Math. Phys.",
    volume = "6",
    pages = "918--919",
    year = "1965"
}

@article{Weller:2024qvo,
    author = "Weller, Colin and Li, Dongjun and Chen, Yanbei",
    title = "{Spectroscopy of bumpy BHs: The nonrotating case}",
    eprint = "2405.20934",
    archivePrefix = "arXiv",
    primaryClass = "gr-qc",
    doi = "10.1103/PhysRevD.111.024064",
    journal = "Phys. Rev. D",
    volume = "111",
    number = "2",
    pages = "024064",
    year = "2025"
}

@article{Cardoso:2016rao,
    author = "Cardoso, Vitor and Franzin, Edgardo and Pani, Paolo",
    title = "{Is the gravitational-wave ringdown a probe of the event horizon?}",
    eprint = "1602.07309",
    archivePrefix = "arXiv",
    primaryClass = "gr-qc",
    doi = "10.1103/PhysRevLett.116.171101",
    journal = "Phys. Rev. Lett.",
    volume = "116",
    number = "17",
    pages = "171101",
    year = "2016",
    note = "[Erratum: Phys.Rev.Lett. 117, 089902 (2016)]"
}

@article{Yunes:2009hc,
    author = "Yunes, Nicolas and Pretorius, Frans",
    title = "{Dynamical Chern-Simons Modified Gravity. I. Spinning Black Holes in the Slow-Rotation Approximation}",
    eprint = "0902.4669",
    archivePrefix = "arXiv",
    primaryClass = "gr-qc",
    doi = "10.1103/PhysRevD.79.084043",
    journal = "Phys. Rev. D",
    volume = "79",
    pages = "084043",
    year = "2009"
}

@article{Regge:1957td,
    author = "Regge, Tullio and Wheeler, John A.",
    title = "{Stability of a Schwarzschild singularity}",
    doi = "10.1103/PhysRev.108.1063",
    journal = "Phys. Rev.",
    volume = "108",
    pages = "1063--1069",
    year = "1957"
}

@Book{Wald:1984cw,
     author    = "Wald, R.M.",
     publisher = "The University of Chicago Press",
     year      = "1984",
     title     = "General Relativity",
     address   = "Chicago"
}

@article{Chrzanowski:1975wv,
    author = "Chrzanowski, P. L.",
    title = "{Vector Potential and Metric Perturbations of a Rotating Black Hole}",
    doi = "10.1103/PhysRevD.11.2042",
    journal = "Phys. Rev. D",
    volume = "11",
    pages = "2042--2062",
    year = "1975"
}

@article{Ayzenberg:2014aka,
    author = "Ayzenberg, Dimitry and Yunes, Nicolas",
    title = "{Slowly-Rotating Black Holes in Einstein-Dilaton-Gauss-Bonnet Gravity: Quadratic Order in Spin Solutions}",
    eprint = "1405.2133",
    archivePrefix = "arXiv",
    primaryClass = "gr-qc",
    doi = "10.1103/PhysRevD.90.044066",
    journal = "Phys. Rev. D",
    volume = "90",
    pages = "044066",
    year = "2014",
    note = "[Erratum: Phys.Rev.D 91, 069905 (2015)]"
}

@article{Chandrasekhar:1975zza,
    author = "Chandrasekhar, S. and Detweiler, Steven L.",
    title = "{The quasi-normal modes of the Schwarzschild black hole}",
    doi = "10.1098/rspa.1975.0112",
    journal = "Proc. Roy. Soc. Lond. A",
    volume = "344",
    pages = "441--452",
    year = "1975"
}

@article{Berti:2003ud,
    author = "Berti, E. and Kokkotas, K. D.",
    title = {{Quasinormal modes of Reissner-Nordstr{\"o}m-anti-de Sitter black holes: Scalar, electromagnetic and gravitational perturbations}},
    eprint = "gr-qc/0301052",
    archivePrefix = "arXiv",
    doi = "10.1103/PhysRevD.67.064020",
    journal = "Phys. Rev. D",
    volume = "67",
    pages = "064020",
    year = "2003"
}

@article{Anderson:1991kx,
    author = "Anderson, A. and Price, R. H.",
    title = "{Intertwining of the equations of black hole perturbations}",
    doi = "10.1103/PhysRevD.43.3147",
    journal = "Phys. Rev. D",
    volume = "43",
    pages = "3147--3154",
    year = "1991"
}

@article{Teukolsky:1973ha,
    author = "Teukolsky, Saul A.",
    title = "{Perturbations of a rotating black hole. 1. Fundamental equations for gravitational electromagnetic and neutrino field perturbations}",
    doi = "10.1086/152444",
    journal = "Astrophys. J.",
    volume = "185",
    pages = "635--647",
    year = "1973"
}

@article{Teukolsky:1972my,
    author = "Teukolsky, S. A.",
    title = "{Rotating black holes - separable wave equations for gravitational and electromagnetic perturbations}",
    reportNumber = "OAP-291",
    doi = "10.1103/PhysRevLett.29.1114",
    journal = "Phys. Rev. Lett.",
    volume = "29",
    pages = "1114--1118",
    year = "1972"
}

@article{Cano:2023tmv,
    author = "Cano, Pablo A. and Fransen, Kwinten and Hertog, Thomas and Maenaut, Simon",
    title = "{Universal Teukolsky equations and black hole perturbations in higher-derivative gravity}",
    eprint = "2304.02663",
    archivePrefix = "arXiv",
    primaryClass = "gr-qc",
    doi = "10.1103/PhysRevD.108.024040",
    journal = "Phys. Rev. D",
    volume = "108",
    number = "2",
    pages = "024040",
    year = "2023"
}

@article{Press:1971wr,
    author = "Press, William H.",
    title = "{Long Wave Trains of Gravitational Waves from a Vibrating Black Hole}",
    reportNumber = "OAP-262",
    doi = "10.1086/180849",
    journal = "Astrophys. J. Lett.",
    volume = "170",
    pages = "L105--L108",
    year = "1971"
}

@article{Silva:2019scu,
    author = "Silva, Hector O. and Glampedakis, Kostas",
    title = "{Eikonal quasinormal modes of black holes beyond general relativity. II. Generalized scalar-tensor perturbations}",
    eprint = "1912.09286",
    archivePrefix = "arXiv",
    primaryClass = "gr-qc",
    doi = "10.1103/PhysRevD.101.044051",
    journal = "Phys. Rev. D",
    volume = "101",
    number = "4",
    pages = "044051",
    year = "2020"
}

@article{Moncrief:1974am,
    author = "Moncrief, V.",
    title = "{Gravitational perturbations of spherically symmetric systems. I. The exterior problem.}",
    doi = "10.1016/0003-4916(74)90173-0",
    journal = "Annals Phys.",
    volume = "88",
    pages = "323--342",
    year = "1974"
}

@article{Vishveshwara:1970zz,
    author = "Vishveshwara, C. V.",
    title = "{Scattering of Gravitational Radiation by a Schwarzschild Black-hole}",
    doi = "10.1038/227936a0",
    journal = "Nature",
    volume = "227",
    pages = "936--938",
    year = "1970"
}

@article{Delsate:2018ome,
    author = "Delsate, Terence and Herdeiro, Carlos and Radu, Eugen",
    title = "{Non-perturbative spinning black holes in dynamical Chern\textendash{}Simons gravity}",
    eprint = "1806.06700",
    archivePrefix = "arXiv",
    primaryClass = "gr-qc",
    doi = "10.1016/j.physletb.2018.09.060",
    journal = "Phys. Lett. B",
    volume = "787",
    pages = "8--15",
    year = "2018"
}

@article{Li:2026rkf,
    author = "Li, Dongjun and Yunes, Nicol{\'a}s",
    title = "{Metric Reconstruction for Generic Black-Hole Perturbations}",
    eprint = "2605.11080",
    archivePrefix = "arXiv",
    primaryClass = "gr-qc",
    month = "5",
    year = "2026"
}

@article{Gibbons:1987ps,
    author = "Gibbons, G. W. and Maeda, Kei-ichi",
    title = "{Black Holes and Membranes in Higher Dimensional Theories with Dilaton Fields}",
    reportNumber = "UTAP-48-87, LPTENS-87-10",
    doi = "10.1016/0550-3213(88)90006-5",
    journal = "Nucl. Phys. B",
    volume = "298",
    pages = "741--775",
    year = "1988"
}

@article{Tattersall:2018axd,
    author = "Tattersall, Oliver J.",
    title = "{Kerr{\textendash}(anti{\textendash})de Sitter black holes: Perturbations and quasinormal modes in the slow rotation limit}",
    eprint = "1808.10758",
    archivePrefix = "arXiv",
    primaryClass = "gr-qc",
    doi = "10.1103/PhysRevD.98.104013",
    journal = "Phys. Rev. D",
    volume = "98",
    number = "10",
    pages = "104013",
    year = "2018"
}

@article{Gerlach:1979rw,
    author = "Gerlach, U. H. and Sengupta, U. K.",
    title = "{GAUGE INVARIANT PERTURBATIONS ON MOST GENERAL SPHERICALLY SYMMETRIC SPACE-TIMES}",
    doi = "10.1103/PhysRevD.19.2268",
    journal = "Phys. Rev. D",
    volume = "19",
    pages = "2268--2272",
    year = "1979"
}

@article{Zerilli:1970wzz,
    author = "Zerilli, F. J.",
    title = "{Gravitational field of a particle falling in a schwarzschild geometry analyzed in tensor harmonics}",
    doi = "10.1103/PhysRevD.2.2141",
    journal = "Phys. Rev. D",
    volume = "2",
    pages = "2141--2160",
    year = "1970"
}

@article{Cano:2024ezp,
    author = {Cano, Pablo A. and Capuano, Lodovico and Franchini, Nicola and Maenaut, Simon and V{\"o}lkel, Sebastian H.},
    title = "{Higher-derivative corrections to the Kerr quasinormal mode spectrum}",
    eprint = "2409.04517",
    archivePrefix = "arXiv",
    primaryClass = "gr-qc",
    doi = "10.1103/PhysRevD.110.124057",
    journal = "Phys. Rev. D",
    volume = "110",
    number = "12",
    pages = "124057",
    year = "2024"
}

@book{Stephani:2003tm,
    author = "Stephani, Hans and Kramer, D. and MacCallum, Malcolm A. H. and Hoenselaers, Cornelius and Herlt, Eduard",
    title = "{Exact solutions of Einstein's field equations}",
    doi = "10.1017/CBO9780511535185",
    isbn = "978-0-521-46702-5, 978-0-511-05917-9",
    publisher = "Cambridge Univ. Press",
    address = "Cambridge",
    series = "Cambridge Monographs on Mathematical Physics",
    year = "2003"
}

@article{Pani:2009ss,
    author = "Pani, Paolo and Berti, Emanuele and Cardoso, Vitor and Chen, Yanbei and Norte, Richard",
    title = "{Gravitational wave signatures of the absence of an event horizon. I. Nonradial oscillations of a thin-shell gravastar}",
    eprint = "0909.0287",
    archivePrefix = "arXiv",
    primaryClass = "gr-qc",
    doi = "10.1103/PhysRevD.80.124047",
    journal = "Phys. Rev. D",
    volume = "80",
    pages = "124047",
    year = "2009"
}

@article{Laeuger:2025zgb,
    author = "Laeuger, Andrew and Weller, Colin and Li, Dongjun and Chen, Yanbei",
    title = "{Ringdown of a black hole surrounded by a thin shell of matter}",
    eprint = "2506.00367",
    archivePrefix = "arXiv",
    primaryClass = "gr-qc",
    doi = "10.1103/gkj4-m1c1",
    journal = "Phys. Rev. D",
    volume = "112",
    number = "8",
    pages = "084042",
    year = "2025"
}

@article{Campanelli:1998jv,
    author = "Campanelli, Manuela and Lousto, Carlos O.",
    title = "{Second order gauge invariant gravitational perturbations of a Kerr black hole}",
    eprint = "gr-qc/9811019",
    archivePrefix = "arXiv",
    reportNumber = "AEI-096",
    doi = "10.1103/PhysRevD.59.124022",
    journal = "Phys. Rev. D",
    volume = "59",
    pages = "124022",
    year = "1999"
}

@article{Detweiler:1980gk,
    author = "Detweiler, Steven L.",
    title = "{BLACK HOLES AND GRAVITATIONAL WAVES. III. THE RESONANT FREQUENCIES OF ROTATING HOLES}",
    doi = "10.1086/158109",
    journal = "Astrophys. J.",
    volume = "239",
    pages = "292--295",
    year = "1980"
}

@article{Pani:2013wsa,
    author = "Pani, Paolo and Berti, Emanuele and Gualtieri, Leonardo",
    title = "{Scalar, Electromagnetic and Gravitational Perturbations of Kerr-Newman Black Holes in the Slow-Rotation Limit}",
    eprint = "1307.7315",
    archivePrefix = "arXiv",
    primaryClass = "gr-qc",
    doi = "10.1103/PhysRevD.88.064048",
    journal = "Phys. Rev. D",
    volume = "88",
    pages = "064048",
    year = "2013"
}

@article{Berti:2020kgk,
    author = "Berti, Emanuele and Collodel, Lucas G. and Kleihaus, Burkhard and Kunz, Jutta",
    title = "{Spin-induced black-hole scalarization in Einstein-scalar-Gauss-Bonnet theory}",
    eprint = "2009.03905",
    archivePrefix = "arXiv",
    primaryClass = "gr-qc",
    doi = "10.1103/PhysRevLett.126.011104",
    journal = "Phys. Rev. Lett.",
    volume = "126",
    number = "1",
    pages = "011104",
    year = "2021"
}

@article{Herdeiro:2020wei,
    author = "Herdeiro, Carlos A. R. and Radu, Eugen and Silva, Hector O. and Sotiriou, Thomas P. and Yunes, Nicol{\'a}s",
    title = "{Spin-induced scalarized black holes}",
    eprint = "2009.03904",
    archivePrefix = "arXiv",
    primaryClass = "gr-qc",
    doi = "10.1103/PhysRevLett.126.011103",
    journal = "Phys. Rev. Lett.",
    volume = "126",
    number = "1",
    pages = "011103",
    year = "2021"
}

@article{Collodel:2019kkx,
    author = "Collodel, Lucas G. and Kleihaus, Burkhard and Kunz, Jutta and Berti, Emanuele",
    title = "{Spinning and excited black holes in Einstein-scalar-Gauss{\textendash}Bonnet theory}",
    eprint = "1912.05382",
    archivePrefix = "arXiv",
    primaryClass = "gr-qc",
    doi = "10.1088/1361-6382/ab74f9",
    journal = "Class. Quant. Grav.",
    volume = "37",
    number = "7",
    pages = "075018",
    year = "2020"
}

@article{Cunha:2019dwb,
    author = "Cunha, Pedro V. P. and Herdeiro, Carlos A. R. and Radu, Eugen",
    title = "{Spontaneously Scalarized Kerr Black Holes in Extended Scalar-Tensor{\textendash}Gauss-Bonnet Gravity}",
    eprint = "1904.09997",
    archivePrefix = "arXiv",
    primaryClass = "gr-qc",
    doi = "10.1103/PhysRevLett.123.011101",
    journal = "Phys. Rev. Lett.",
    volume = "123",
    number = "1",
    pages = "011101",
    year = "2019"
}

@article{Cardoso:2002pa,
    author = "Cardoso, Vitor and Dias, Oscar J. C. and Lemos, Jose P. S.",
    title = "{Gravitational radiation in D-dimensional space-times}",
    eprint = "hep-th/0212168",
    archivePrefix = "arXiv",
    doi = "10.1103/PhysRevD.67.064026",
    journal = "Phys. Rev. D",
    volume = "67",
    pages = "064026",
    year = "2003"
}

@article{Chung:2024ira,
    author = "Chung, Adrian Ka-Wai and Yunes, Nicolas",
    title = "{Ringing out General Relativity: Quasinormal Mode Frequencies for Black Holes of Any Spin in Modified Gravity}",
    eprint = "2405.12280",
    archivePrefix = "arXiv",
    primaryClass = "gr-qc",
    doi = "10.1103/PhysRevLett.133.181401",
    journal = "Phys. Rev. Lett.",
    volume = "133",
    number = "18",
    pages = "181401",
    year = "2024"
}

@article{Dima:2020yac,
    author = "Dima, Alexandru and Barausse, Enrico and Franchini, Nicola and Sotiriou, Thomas P.",
    title = "{Spin-induced black hole spontaneous scalarization}",
    eprint = "2006.03095",
    archivePrefix = "arXiv",
    primaryClass = "gr-qc",
    doi = "10.1103/PhysRevLett.125.231101",
    journal = "Phys. Rev. Lett.",
    volume = "125",
    number = "23",
    pages = "231101",
    year = "2020"
}

@article{Chung:2025gyg,
    author = "Chung, Adrian Ka-Wai and Lam, Kelvin Ka-Ho and Yunes, Nicolas",
    title = "{Quasinormal mode frequencies and gravitational perturbations of spinning black holes in modified gravity through METRICS: The dynamical Chern-Simons gravity case}",
    eprint = "2503.11759",
    archivePrefix = "arXiv",
    primaryClass = "gr-qc",
    doi = "10.1103/g83n-rrlj",
    journal = "Phys. Rev. D",
    volume = "111",
    number = "12",
    pages = "124052",
    year = "2025"
}

@article{Johannsen:2011dh,
    author = "Johannsen, Tim and Psaltis, Dimitrios",
    title = "{A Metric for Rapidly Spinning Black Holes Suitable for Strong-Field Tests of the No-Hair Theorem}",
    eprint = "1105.3191",
    archivePrefix = "arXiv",
    primaryClass = "gr-qc",
    doi = "10.1103/PhysRevD.83.124015",
    journal = "Phys. Rev. D",
    volume = "83",
    pages = "124015",
    year = "2011"
}

@article{Berens:2025kkm,
    author = "Berens, Roman and Gravely, Trevor and Lupsasca, Alexandru",
    title = "{Gravitational waves on Kerr black holes: II. Metric reconstruction with cosmological constant}",
    eprint = "2510.07712",
    archivePrefix = "arXiv",
    primaryClass = "gr-qc",
    doi = "10.1088/1361-6382/ae4282",
    journal = "Class. Quant. Grav.",
    volume = "43",
    number = "4",
    pages = "045007",
    year = "2026"
}

@incollection{PETROV196988,
title = {CHAPTER 3 - GENERAL CLASSIFICATION OF GRAVITATIONAL FIELDS},
booktitle = {Einstein Spaces},
publisher = {Pergamon},
pages = {88-131},
year = {1969},
isbn = {978-0-08-012315-8},
doi = {https://doi.org/10.1016/B978-0-08-012315-8.50008-2},
url = {https://www.sciencedirect.com/science/article/pii/B9780080123158500082},
author = {A.Z. PETROV}
}

@article{Wardell:2024yoi,
    author = "Wardell, Barry and Kavanagh, Chris and Dolan, Sam R.",
    title = "{Sourced metric perturbations of Kerr spacetime in Lorenz gauge}",
    eprint = "2406.12510",
    archivePrefix = "arXiv",
    primaryClass = "gr-qc",
    doi = "10.1088/1361-6382/ae0918",
    journal = "Class. Quant. Grav.",
    volume = "42",
    number = "20",
    pages = "205007",
    year = "2025"
}

@article{Li:2025ffh,
    author = "Li, Dongjun and Weller, Colin and Bourg, Patrick and LaHaye, Michael and Yunes, Nicol{\'a}s and Yang, Huan",
    title = "{Extreme mass-ratio inspiral within an ultralight scalar cloud: Scalar radiation}",
    eprint = "2507.02045",
    archivePrefix = "arXiv",
    primaryClass = "gr-qc",
    doi = "10.1103/7l9s-g21j",
    journal = "Phys. Rev. D",
    volume = "112",
    number = "8",
    pages = "084057",
    year = "2025"
}

@article{Yagi:2012ya,
    author = "Yagi, Kent and Yunes, Nicolas and Tanaka, Takahiro",
    title = "{Slowly Rotating Black Holes in Dynamical Chern-Simons Gravity: Deformation Quadratic in the Spin}",
    eprint = "1206.6130",
    archivePrefix = "arXiv",
    primaryClass = "gr-qc",
    doi = "10.1103/PhysRevD.86.044037",
    journal = "Phys. Rev. D",
    volume = "86",
    pages = "044037",
    year = "2012",
    note = "[Erratum: Phys.Rev.D 89, 049902 (2014)]"
}

@article{Hussain:2022ins,
    author = "Hussain, Asad and Zimmerman, Aaron",
    title = "{Approach to computing spectral shifts for black holes beyond Kerr}",
    eprint = "2206.10653",
    archivePrefix = "arXiv",
    primaryClass = "gr-qc",
    doi = "10.1103/PhysRevD.106.104018",
    journal = "Phys. Rev. D",
    volume = "106",
    number = "10",
    pages = "104018",
    year = "2022"
}

@article{Press:1973zz,
    author = "Press, William H. and Teukolsky, Saul A.",
    title = "{Perturbations of a Rotating Black Hole. II. Dynamical Stability of the Kerr Metric}",
    doi = "10.1086/152445",
    journal = "Astrophys. J.",
    volume = "185",
    pages = "649--674",
    year = "1973"
}

@article{Kleihaus:2015aje,
    author = "Kleihaus, Burkhard and Kunz, Jutta and Mojica, Sindy and Radu, Eugen",
    title = "{Spinning black holes in Einstein{\textendash}Gauss-Bonnet{\textendash}dilaton theory: Nonperturbative solutions}",
    eprint = "1511.05513",
    archivePrefix = "arXiv",
    primaryClass = "gr-qc",
    doi = "10.1103/PhysRevD.93.044047",
    journal = "Phys. Rev. D",
    volume = "93",
    number = "4",
    pages = "044047",
    year = "2016"
}

@article{Cano:2024wzo,
    author = "Cano, Pablo A. and David, Marina",
    title = "{Isospectrality in Effective Field Theory Extensions of General Relativity}",
    eprint = "2407.12080",
    archivePrefix = "arXiv",
    primaryClass = "hep-th",
    doi = "10.1103/PhysRevLett.134.191401",
    journal = "Phys. Rev. Lett.",
    volume = "134",
    number = "19",
    pages = "191401",
    year = "2025"
}

@article{Kleihaus:2011tg,
    author = "Kleihaus, Burkhard and Kunz, Jutta and Radu, Eugen",
    title = "{Rotating Black Holes in Dilatonic Einstein-Gauss-Bonnet Theory}",
    eprint = "1101.2868",
    archivePrefix = "arXiv",
    primaryClass = "gr-qc",
    doi = "10.1103/PhysRevLett.106.151104",
    journal = "Phys. Rev. Lett.",
    volume = "106",
    pages = "151104",
    year = "2011"
}

@article{Zhao:2023jiz,
    author = "Zhao, Ying and Liu, Wentao and Zhang, Chao and Fang, Xiongjun and Jing, Jiliang",
    title = "{Quasinormal modes and isospectrality of Bardeen (Anti-) de Sitter black holes*}",
    eprint = "2306.02332",
    archivePrefix = "arXiv",
    primaryClass = "gr-qc",
    doi = "10.1088/1674-1137/ad1ed8",
    journal = "Chin. Phys. C",
    volume = "48",
    number = "3",
    pages = "035102",
    year = "2024"
}

@article{Fernando:2016ksb,
    author = "Fernando, Sharmanthie",
    title = "{Bardeen{\textendash}de Sitter black holes}",
    eprint = "1611.05337",
    archivePrefix = "arXiv",
    primaryClass = "gr-qc",
    reportNumber = "NKU-2016-SF2",
    doi = "10.1142/S0218271817500717",
    journal = "Int. J. Mod. Phys. D",
    volume = "26",
    number = "07",
    pages = "1750071",
    year = "2017"
}

@article{Liu:2024oeq,
    author = "Liu, Wentao and Fang, Xiongjun and Jing, Jiliang and Wang, Jieci",
    title = "{Lorentz violation induces isospectrality breaking in Einstein-bumblebee gravity theory}",
    eprint = "2402.09686",
    archivePrefix = "arXiv",
    primaryClass = "gr-qc",
    doi = "10.1007/s11433-024-2405-y",
    journal = "Sci. China Phys. Mech. Astron.",
    volume = "67",
    number = "8",
    pages = "280413",
    year = "2024"
}

@article{Maluf:2020kgf,
    author = "Maluf, R. V. and Neves, Juliano C. S.",
    title = "{Black holes with a cosmological constant in bumblebee gravity}",
    eprint = "2011.12841",
    archivePrefix = "arXiv",
    primaryClass = "gr-qc",
    doi = "10.1103/PhysRevD.103.044002",
    journal = "Phys. Rev. D",
    volume = "103",
    number = "4",
    pages = "044002",
    year = "2021"
}

@article{LuisBlazquez-Salcedo:2020rqp,
    author = "Luis Bl{\'a}zquez-Salcedo, Jose and Herdeiro, Carlos A. R. and Kahlen, Sarah and Kunz, Jutta and Pombo, Alexandre M. and Radu, Eugen",
    title = "{Quasinormal modes of hot, cold and bald Einstein{\textendash}Maxwell-scalar black holes}",
    eprint = "2008.11744",
    archivePrefix = "arXiv",
    primaryClass = "gr-qc",
    doi = "10.1140/epjc/s10052-021-08952-w",
    journal = "Eur. Phys. J. C",
    volume = "81",
    number = "2",
    pages = "155",
    year = "2021"
}

@article{Blazquez-Salcedo:2020nhs,
    author = "Bl{\'a}zquez-Salcedo, Jose Luis and Herdeiro, Carlos A. R. and Kunz, Jutta and Pombo, Alexandre M. and Radu, Eugen",
    title = "{Einstein-Maxwell-scalar black holes: the hot, the cold and the bald}",
    eprint = "2002.00963",
    archivePrefix = "arXiv",
    primaryClass = "gr-qc",
    doi = "10.1016/j.physletb.2020.135493",
    journal = "Phys. Lett. B",
    volume = "806",
    pages = "135493",
    year = "2020"
}

@article{Pursey:1986kk,
  title = {Isometric operators, isospectral Hamiltonians, and supersymmetric quantum mechanics},
  author = {Pursey, D. L.},
  journal = {Phys. Rev. D},
  volume = {33},
  issue = {8},
  pages = {2267--2279},
  numpages = {0},
  year = {1986},
  month = {Apr},
  publisher = {American Physical Society},
  doi = {10.1103/PhysRevD.33.2267},
  url = {https://link.aps.org/doi/10.1103/PhysRevD.33.2267}
}

@book{trachanas2009exactly,
  author    = {Trachanas, S.},
  title     = {Exactly Solvable Quantum Mechanical Potentials},
  publisher = {Crete University Press},
  address   = {Heraklion, Greece},
  year      = {2009}
}

@article{Glampedakis:2017rar,
    author = "Glampedakis, Kostas and Johnson, Aaron D. and Kennefick, Daniel",
    title = "{Darboux transformation in black hole perturbation theory}",
    eprint = "1702.06459",
    archivePrefix = "arXiv",
    primaryClass = "gr-qc",
    doi = "10.1103/PhysRevD.96.024036",
    journal = "Phys. Rev. D",
    volume = "96",
    number = "2",
    pages = "024036",
    year = "2017"
}

@article{Pereniguez:2026avs,
    author = "Pere{\~n}iguez, David",
    title = "{Unifying the Regge-Wheeler-Zerilli and Bardeen-Press-Teukolsky formalisms on spherical backgrounds}",
    eprint = "2605.04147",
    archivePrefix = "arXiv",
    primaryClass = "gr-qc",
    month = "5",
    year = "2026"
}

@article{Cardoso:2019mqo,
    author = "Cardoso, Vitor and Kimura, Masashi and Maselli, Andrea and Berti, Emanuele and Macedo, Caio F. B. and McManus, Ryan",
    title = "{Parametrized black hole quasinormal ringdown: Decoupled equations for nonrotating black holes}",
    eprint = "1901.01265",
    archivePrefix = "arXiv",
    primaryClass = "gr-qc",
    doi = "10.1103/PhysRevD.99.104077",
    journal = "Phys. Rev. D",
    volume = "99",
    number = "10",
    pages = "104077",
    year = "2019"
}

@article{Chung:2024vaf,
    author = "Chung, Adrian Ka-Wai and Yunes, Nicolas",
    title = "{Quasinormal mode frequencies and gravitational perturbations of black holes with any subextremal spin in modified gravity through METRICS: The scalar-Gauss-Bonnet gravity case}",
    eprint = "2406.11986",
    archivePrefix = "arXiv",
    primaryClass = "gr-qc",
    doi = "10.1103/PhysRevD.110.064019",
    journal = "Phys. Rev. D",
    volume = "110",
    number = "6",
    pages = "064019",
    year = "2024"
}

@article{Martel:2005ir,
    author = "Martel, Karl and Poisson, Eric",
    title = "{Gravitational perturbations of the Schwarzschild spacetime: A Practical covariant and gauge-invariant formalism}",
    eprint = "gr-qc/0502028",
    archivePrefix = "arXiv",
    doi = "10.1103/PhysRevD.71.104003",
    journal = "Phys. Rev. D",
    volume = "71",
    pages = "104003",
    year = "2005"
}

@article{Loutrel:2020wbw,
    author = "Loutrel, Nicholas and Ripley, Justin L. and Giorgi, Elena and Pretorius, Frans",
    title = "{Second Order Perturbations of Kerr Black Holes: Reconstruction of the Metric}",
    eprint = "2008.11770",
    archivePrefix = "arXiv",
    primaryClass = "gr-qc",
    doi = "10.1103/PhysRevD.103.104017",
    journal = "Phys. Rev. D",
    volume = "103",
    number = "10",
    pages = "104017",
    year = "2021"
}

@book{Chandrasekhar_1983, 
	place={Oxford [Oxfordshire]: New York}, 
 	series={The International series of monographs on physics}, 
 	title={The mathematical theory of black holes}, 
 	ISBN={978-0-19-851291-2}, 
 	publisher={Clarendon Press}, 
 	author={Chandrasekhar, S.}, 
 	year={1983}, 
 	collection={The International series of monographs on physics}
}

@article{Zerilli:1970se,
    author = "Zerilli, Frank J.",
    title = "{Effective potential for even parity Regge-Wheeler gravitational perturbation equations}",
    doi = "10.1103/PhysRevLett.24.737",
    journal = "Phys. Rev. Lett.",
    volume = "24",
    pages = "737--738",
    year = "1970"
}

@article{Andersson:2015xla,
    author = {Andersson, Lars and B{\"a}ckdahl, Thomas and Blue, Pieter},
    title = "{Spin geometry and conservation laws in the Kerr spacetime}",
    eprint = "1504.02069",
    archivePrefix = "arXiv",
    primaryClass = "gr-qc",
    doi = "10.4310/sdg.2015.v20.n1.a8",
    journal = "Surveys Diff. Geom.",
    volume = "20",
    number = "1",
    pages = "183--226",
    year = "2015"
}

@article{Dunajski:2009dqa,
    author = "Dunajski, Maciej and Tod, Paul",
    title = "{Four--Dimensional Metrics Conformal to Kahler}",
    eprint = "0901.2261",
    archivePrefix = "arXiv",
    primaryClass = "math.DG",
    reportNumber = "DAMTP-2009-3",
    doi = "10.1017/S030500410999048X",
    journal = "Math. Proc. Cambridge Phil. Soc.",
    volume = "148",
    pages = "485",
    year = "2010"
}

@article{Araneda:2025lak,
    author = "Araneda, Bernardo and Dunajski, Maciej",
    title = "{New Asymptotically Flat Einstein-Maxwell Instantons}",
    eprint = "2510.06458",
    archivePrefix = "arXiv",
    primaryClass = "gr-qc",
    doi = "10.1103/f3ls-znl6",
    journal = "Phys. Rev. Lett.",
    volume = "135",
    number = "24",
    pages = "241501",
    year = "2025"
}

@phdthesis{Floyd:1973,
  author = {Floyd, R.},
  title  = {The Dynamics of Kerr Fields},
  school = {London University},
  year   = {1973}
}

@article{Penrose:1973um,
    author = "Penrose, R.",
    title = "{Naked singularities}",
    doi = "10.1111/j.1749-6632.1973.tb41447.x",
    journal = "Annals N. Y. Acad. Sci.",
    volume = "224",
    pages = "125--134",
    year = "1973"
}

@article{Cano:2024jkd,
    author = {Cano, Pablo A. and Capuano, Lodovico and Franchini, Nicola and Maenaut, Simon and V{\"o}lkel, Sebastian H.},
    title = "{Parametrized quasinormal mode framework for modified Teukolsky equations}",
    eprint = "2407.15947",
    archivePrefix = "arXiv",
    primaryClass = "gr-qc",
    doi = "10.1103/PhysRevD.110.104007",
    journal = "Phys. Rev. D",
    volume = "110",
    number = "10",
    pages = "104007",
    year = "2024",
    note = "[Erratum: Phys.Rev.D 113, 069902 (2026)]"
}

@article{Suvorov:2019qow,
    author = "Suvorov, Arthur G.",
    title = "{Gravitational perturbations of a Kerr black hole in $f(R)$ gravity}",
    eprint = "1905.02021",
    archivePrefix = "arXiv",
    primaryClass = "gr-qc",
    doi = "10.1103/PhysRevD.99.124026",
    journal = "Phys. Rev. D",
    volume = "99",
    number = "12",
    pages = "124026",
    year = "2019"
}

@article{Monteiro:2014cda,
    author = "Monteiro, Ricardo and O'Connell, Donal and White, Chris D.",
    title = "{Black holes and the double copy}",
    eprint = "1410.0239",
    archivePrefix = "arXiv",
    primaryClass = "hep-th",
    reportNumber = "EDINBURGH-2014-18",
    doi = "10.1007/JHEP12(2014)056",
    journal = "JHEP",
    volume = "12",
    pages = "056",
    year = "2014"
}

@article{Kent:2025pvu,
    author = "Kent, Brian and Zimmerman, Aaron",
    title = "{New Framework for Classical Double Copies}",
    eprint = "2505.03887",
    archivePrefix = "arXiv",
    primaryClass = "hep-th",
    doi = "10.1103/xn1j-ddcc",
    journal = "Phys. Rev. Lett.",
    volume = "135",
    number = "14",
    pages = "141501",
    year = "2025"
}

@article{Yoshida:2010zzb,
    author = "Yoshida, Shijun and Uchikata, Nami and Futamase, Toshifumi",
    title = "{Quasinormal modes of Kerr-de Sitter black holes}",
    doi = "10.1103/PhysRevD.81.044005",
    journal = "Phys. Rev. D",
    volume = "81",
    pages = "044005",
    year = "2010"
}

@article{Kubiznak:2007kh,
    author = "Kubiznak, David and Krtous, Pavel",
    title = "{On conformal Killing-Yano tensors for Plebanski-Demianski family of solutions}",
    eprint = "0707.0409",
    archivePrefix = "arXiv",
    primaryClass = "gr-qc",
    doi = "10.1103/PhysRevD.76.084036",
    journal = "Phys. Rev. D",
    volume = "76",
    pages = "084036",
    year = "2007"
}

@article{Berti:2016lat,
    author = "Berti, Emanuele and Sesana, Alberto and Barausse, Enrico and Cardoso, Vitor and Belczynski, Krzysztof",
    title = "{Spectroscopy of Kerr black holes with Earth- and space-based interferometers}",
    eprint = "1605.09286",
    archivePrefix = "arXiv",
    primaryClass = "gr-qc",
    doi = "10.1103/PhysRevLett.117.101102",
    journal = "Phys. Rev. Lett.",
    volume = "117",
    number = "10",
    pages = "101102",
    year = "2016"
}

@article{Isi:2020tac,
    author = "Isi, Maximiliano and Farr, Will M. and Giesler, Matthew and Scheel, Mark A. and Teukolsky, Saul A.",
    title = "{Testing the Black-Hole Area Law with GW150914}",
    eprint = "2012.04486",
    archivePrefix = "arXiv",
    primaryClass = "gr-qc",
    reportNumber = "LIGO-P2000507",
    doi = "10.1103/PhysRevLett.127.011103",
    journal = "Phys. Rev. Lett.",
    volume = "127",
    number = "1",
    pages = "011103",
    year = "2021"
}

@article{Green:2026nlt,
    author = "Green, Stephen R. and Krasnov, Kirill and Shaw, Adam",
    title = "{Kahler decoupling for Kerr perturbations}",
    eprint = "2604.22424",
    archivePrefix = "arXiv",
    primaryClass = "gr-qc",
    month = "4",
    year = "2026"
}

@article{Cano:2019ore,
    author = "Cano, Pablo A. and Ruip{\'e}rez, Alejandro",
    title = "{Leading higher-derivative corrections to Kerr geometry}",
    eprint = "1901.01315",
    archivePrefix = "arXiv",
    primaryClass = "gr-qc",
    reportNumber = "IFT-UAM/CSIC-19-2",
    doi = "10.1007/JHEP05(2019)189",
    journal = "JHEP",
    volume = "05",
    pages = "189",
    year = "2019",
    note = "[Erratum: JHEP 03, 187 (2020)]"
}

@article{R:2022tqa,
    author = "R, Abhishek Hegade K. and Most, Elias R. and Noronha, Jorge and Witek, Helvi and Yunes, Nicol{\'a}s",
    title = "{How do axisymmetric black holes grow monopole and dipole hair?}",
    eprint = "2212.02039",
    archivePrefix = "arXiv",
    primaryClass = "gr-qc",
    doi = "10.1103/PhysRevD.107.104047",
    journal = "Phys. Rev. D",
    volume = "107",
    number = "10",
    pages = "104047",
    year = "2023"
}

@article{Mukkamala:2024dxf,
    author = "Mukkamala, Gowtham Rishi and Pere{\~n}iguez, David",
    title = "{Decoupled gravitational wave equations in spherical symmetry from curvature wave equations}",
    eprint = "2408.13557",
    archivePrefix = "arXiv",
    primaryClass = "gr-qc",
    doi = "10.1088/1475-7516/2025/01/122",
    journal = "JCAP",
    volume = "01",
    pages = "122",
    year = "2025"
}
\appendix


\newpage
\onecolumngrid

\section{End Matter}
We use the standard Newman--Penrose (NP) null tetrad $\{l^{\mu}, n^{\mu},m^{\mu}, \bar{m}^{\mu}\}$, where $g_{\mu \nu}l^{\mu}n^{\nu} = -1$ and $g_{\mu\nu}m^\mu\bar m^\nu=+1$.
The four-vectors $l^{\mu}$ and $n^{\mu}$ are real-valued and the tetrad has the corresponding directional derivatives $D = l^{\mu}\nabla_{\mu}$, $\Delta = n^{\mu}\nabla_{\mu}$, $\delta = m^{\mu}\nabla_{\mu}$, and $\bar{\delta} = \bar{m}^{\mu}\nabla_{\mu}$. We make use of the following operators introduced in~\cite{Wagle:2023fwl}:
\begin{subequations}\label{eq: ModifiedDirectionDeriv}
    \begin{align}
D_{[a, b, c, d]} & =D+a \varepsilon+b \bar{\varepsilon}+c \rho+d \bar{\rho}\,, \\
\Delta_{[a, b, c, d]} & =\Delta+a \mu+b \bar{\mu}+c \gamma+d \bar{\gamma}\,, \\
\delta_{[a, b, c, d]} & =\delta+a \bar{\alpha}+b \beta+c \bar{\pi}+d \tau\,, \\
\bar{\delta}_{[a, b, c, d]} & =\bar{\delta}+a \alpha+b \bar{\beta}+c \pi+d \bar{\tau}\,.
    \end{align}
\end{subequations}
where $a,b,c$, and $d$ are numerical coefficients.

The Teukolsky operators of the Weyl tensor are given by
\begin{subequations}\label{eq: TekOperators}
    \begin{align}
        \mathcal{O}_0 &=D_{[0,0,-4,-1]}\Delta_{[1,0,-4,0]} - \delta_{[-1,-3,1,-4]}\bar{\delta}_{[-4,0,1,0]}- 3 \Psi_2\,,\\
        \mathcal{O}_4 &= \Delta_{[4,1,3,-1]}D_{[0,0,-1,0]} - \bar{\delta}_{[3,1,4,-1]}\delta_{[0,4,0,-1]} - 3 \Psi_2\,,
    \end{align}
\end{subequations}
for the $s = +2$ and $s = -2$ components, respectively. We have assumed a frame and coordinate gauge at $\mathcal{O}(\epsilon^0)$ satisfying
\begin{subequations}\label{eq: FrameGuageSupp}
\begin{align}
    \Psi_0 =  \Psi_1 =  \Psi_3 =  \Psi_4 = 0\,,\\
    \kappa = \sigma = \lambda = \nu = 0\,.
\end{align}
\end{subequations}
The two Bianchi identities used to derive the Teukolsky equation for $\Psi^{(1)}_0$, given in Eq.~\eqref{eq: RicciFormCurved} in the main text, are given by
\begin{subequations}
    \begin{align}
        F_1 \Psi^{(1)}_0 +F_2 \Psi^{(1)}_1   + 3\kappa^{(1)} \Psi^{(0)}_2&= \mathcal{S}_1\Phi^{(1)}_{00} + \mathcal{S}_2\Phi^{(1)}_{01} -2 \kappa^{(1)} \Phi^{(0)}_{11}\,,\label{eq: Bianchi1Supp}\\
        F_3 \Psi^{(1)}_0 +F_4 \Psi^{(1)}_1   + 3\sigma^{(1)} \Psi^{(0)}_2 &= \mathcal{S}_3\Phi^{(1)}_{01} + \mathcal{S}_4\Phi^{(1)}_{02} + 2 \sigma^{(1)}\Phi^{(0)}_{11}\,,\label{eq: Bianchi2Supp}
    \end{align}
\end{subequations}
where the operators $F_i$ and $\mathcal{S}_i$ are given by
\begin{equation}\label{eq:FS-operators}
\begin{alignedat}{4}
F_1 &= \bar{\delta}_{[-4,0,1,0]}, 
&\qquad
F_2 &= D_{[-2,0,-4,0]},
&\qquad
\mathcal{S}_1 &= \delta_{[-2,-2,1,0]},
&\qquad
\mathcal{S}_2 &= -D_{[-2,0,0,-2]},
\\
F_3 &= \Delta_{[1,0,-4,0]},
&\qquad
F_4 &= -\delta_{[0,-2,0,-4]},
&\qquad
\mathcal{S}_3 &= \delta_{[0,-2,2,0]},
&\qquad
\mathcal{S}_4 &= -D_{[-2,2,0,-1]} .
\end{alignedat}
\end{equation}
The perturbed Newman--Penrose quantities $\{\sigma^{(1)},\kappa^{(1)},\Psi^{(1)}_1\}$ can be expressed directly in terms of metric components, but these quantities are not tetrad-gauge invariant so we fix our frame gauge (see Refs.\cite{Campanelli:1998jv,Ripley:2020xby,Li:2026rkf}) to satisfy
\begin{align}\label{eq: CCKFrame}
        l^{\mu(1)}=\frac{1}{2} h_{l l}^{(1)} n^\mu\,,\quad
         n^{\mu(1)}=\frac{1}{2} h_{n n}^{(1)} l^\mu+h_{l n}^{(1)} n^\mu\,, \quad m^{\mu(1)}=h_{n m}^{(1)} l^\mu+h_{l m}^{(1)} n^\mu-\frac{1}{2} h_{m \bar{m}}^{(1)} m^\mu-\frac{1}{2} h_{m m}^{(1)} \bar{m}^\mu\,.
    \end{align}
In the ingoing radiation gauge, where $h_{\mu \nu}l^{\mu} = 0$ and $h = g^{\mu \nu}h_{\mu \nu} = 0$, Eq.~\eqref{eq: CCKFrame} reduces to 
\begin{align}\label{eq: CCKFrameIRG}
        l^{\mu(1)}=0\,,\quad
         n^{\mu(1)}=\frac{1}{2} h_{n n}^{(1)} l^\mu\,,\quad m^{\mu(1)}&=h_{n m}^{(1)} l^\mu-\frac{1}{2} h_{m m}^{(1)} \bar{m}^\mu\,.
    \end{align}
Using the frame gauge of Eq.~\eqref{eq: CCKFrame}, we find the following expressions for the perturbed spin connection:
\begin{subequations}\label{eq: ReconstructedIRGQuants}
    \begin{align}   
\kappa^{(1)} & =\frac{1}{2} \delta_{[-2,-2,1,1]} h_{l l}^{(1)}-D_{[-2,0,0,-1]} h_{l m}^{(1)}\, \\
\sigma^{(1)} & =-\frac{1}{2} D_{[-2,2,1,-1]} h_{m m}^{(1)}+(\bar{\pi}+\tau) h_{l m}^{(1)}\,\\
\Psi_1^{(1)}&=-\frac{1}{8}\left[2 D_{[-1,1,1,-1]} D_{[0,2,1,-1]} h_{n m}^{(1)}+D_{[-1,1,1,-1]} \delta_{[-2,2,-2,-1]} h_{m m}^{(1)}+\bar{\delta}_{[-3,1,-3,-1]} D_{[-2,2,0,-1]} h_{m m}^{(1)}\right],
    \end{align}
\end{subequations}
Transforming the coordinate gauge to the ingoing radiation gauge, these expressions become
 \begin{subequations}\label{eq: ReconstructedIRGQuantsReduced}
    \begin{align}   
\kappa^{(1)} & =0\, \\
\sigma^{(1)} & =-\frac{1}{2} D_{[-2,2,1,-1]} h_{m m}^{(1)}\,\\
\Psi_1^{(1)}&=-\frac{1}{8}\left[2 D_{[-1,1,1,-1]} D_{[0,2,1,-1]} h_{n m}^{(1)}+D_{[-1,1,1,-1]} \delta_{[-2,2,-2,-1]} h_{m m}^{(1)}+\bar{\delta}_{[-3,1,-3,-1]} D_{[-2,2,0,-1]} h_{m m}^{(1)}\right],
    \end{align}
\end{subequations}
Similarly, the perturbed spin $\pm 2$ Weyl scalars can be expressed in terms of the metric components $h_{ab}$ via
\begin{subequations}
    \begin{align}
\Psi^{(1)}_0&=D_{[-3,1,-1,-1]} \left(-\frac{1}{2} D_{[-2,2,1,-1]} h^{(1)}_{m m}+(\bar{\pi}+\tau) h^{(1)}_{l m}\right)-\delta_{[-1,-3,1,-1]} \left(\frac{1}{2} \delta_{[-2,-2,1,1]} h^{(1)}_{l l}-D_{[-2,0,0,-1]} h^{(1)}_{l m}\right)\,,\label{eq: Psi_0MetricComp}\\
\Psi^{(1)}_4&=\bar{\delta}_{[3,1,1,-1]} \left(-\frac{1}{2} \bar{\delta}_{[2,2,-1,-1]} h^{(1)}_{n n}+\Delta_{[0,1,2,0]} h^{(1)}_{n \bar{m}}\right)-\Delta_{[1,1,3,-1]} \left((\pi+\bar{\tau}) h^{(1)}_{n \bar{m}}+\frac{1}{2} \Delta_{[-1,1,2,-2]} h^{(1)}_{\bar{m} \bar{m}}\right)\,.\label{eq: Psi_4MetricComp}
    \end{align}
\end{subequations}
Since $\Psi^{(1)}_0$ and $\Psi^{(1)}_4$ are both coordinate- and tetrad-gauge invariant at $\mathcal{O}(\epsilon^1)$, these expressions do not depend on the choice of gauge at $\mathcal{O}(\epsilon^1)$.

The sourced Teukolsky equations are given by 
\begin{subequations}\label{eq: SourcedTekEqus}
\begin{align}
\mathcal{O}_0 \Psi_0^{(1)} & =\mathscr{T}^{(1)}_0\,, \\
\mathcal{O}_4 \Psi_4^{(1)} & =\mathscr{T}^{(1)}_4\,,
\end{align}
\end{subequations}
where the effective sources are given by 
\begin{subequations}
    \begin{align}
\mathscr{T}^{(1)}_0 &= \delta_{[-1,-3,1,-4]}\left( D_{[-2,0,0,-2]}\Phi^{(1)}_{01}
        - \delta_{[-2,-2,1,0]}\,\Phi^{(1)}_{00} \right)  + D_{[-3,1,-4,-1]}\!\left(\, \delta_{[0,-2,2,0]}\Phi^{(1)}_{01}
        - D_{[-2,2,0,-1]}\Phi^{(1)}_{02} \,\right)\,\\
\mathscr{T}^{(1)}_4 &= \Delta_{[4,1,3,-1]}\left(\, \bar\delta_{[2,0,0,-2]}\,\Phi^{(1)}_{21}
        - \Delta_{[0,1,2,-2]}\Phi^{(1)}_{20} \right) + \bar\delta_{[3,1,4,-1]}\left( \Delta_{[0,2,2,0]}\,\Phi^{(1)}_{21}
        - \bar\delta_{[2,2,0,-1]}\Phi^{(1)}_{22} \right)
    \end{align}
\end{subequations}
The homogeneous Teukolsky equations are recovered whenever $\mathscr{T}^{(1)}_4 = 0 = \mathscr{T}^{(1)}_0$.

\section{Supplemental Material}

\begin{table}[htb]
\captionsetup{justification=raggedright,singlelinecheck=false}
\caption{
Classification of representative theories and spacetimes by their isospectral properties, based on Theorems~\ref{thm: MainTheorem} and \ref{thm: TheoremB} and Lemmas~\ref{lem: ChiralLemma} and \ref{lem: ChiralLemmaParity}. Here, $R$ is the Ricci scalar, $\Lambda$ is the cosmological constant, and $S_{\mathrm{M}}$ and $S_{\mathrm{shell}}$ denote generic matter and thin-shell actions, respectively. The Johannsen--Psaltis and linear, static bumpy BHs entries are phenomenological: depending on the field equations imposed, they can be interpreted either as vacuum beyond-GR deformations or as non-vacuum GR spacetimes; in this table we adopt the latter interpretation and are placed in the GR-with-matter row. The electromagnetic field strength is $F_{\mu\nu}=2\nabla_{[\mu}A_{\nu]}$, with $F^2\equiv F_{\mu\nu}F^{\mu\nu}$, and $\mathcal{L}(F^2)$ is a nonlinear-electrodynamics Lagrangian \cite{Zhao:2023jiz}. The scalar notation is theory-dependent: $\phi$ is a generic scalar field in Einstein--Maxwell--scalar theory, the dilaton in Einstein--Maxwell--dilaton--axion theory; $\chi$ is the axion field; and $\varphi$ is the dilaton in Einstein--dilaton--Gauss--Bonnet gravity, a generic scalar in scalar--Gauss--Bonnet gravity. The functions $f(\phi)$ and $f(\varphi)$ specify the corresponding scalar couplings. The quantity $\mathcal{G}$ is the Gauss--Bonnet invariant \cite{Li:2023ulk}, ${}^{*}RR$ is the Pontryagin density \cite{Li:2023ulk}, $\mathcal{R}^3$ schematically denotes parity-even cubic contractions of the curvature, $\ell$ is the higher-curvature length scale, and $\lambda_{\mathrm{ev}}$ is its parity-even coupling~\cite{Cano:2019ore,R:2022tqa}. The symbols $\alpha$ and $b$ in the scalar-curvature and axion terms denote the associated coupling constants. In the Einstein--Bumblebee row, $B_\mu$ is the bumblebee vector, $B_{\mu\nu}=2\nabla_{[\mu}B_{\nu]}$ is its field strength, $B^2\equiv B_\mu B^\mu$, $b$ fixes the vacuum norm of $B_\mu$, $\zeta$ is the nonminimal curvature coupling, and $f(B^2\pm b^2)$ is the symmetry-breaking potential~\cite{Maluf:2020kgf}.}
\label{tab:isospectrality_examples_reorganized}
\begingroup
\scriptsize
\setlength{\tabcolsep}{3.5pt}
\renewcommand{\arraystretch}{1.25}
\begin{ruledtabular}
\begin{tabular}{l l l c l}
\textbf{Theory} & \textbf{Lagrangian density} & \textbf{Black hole solution} & \textbf{Isospectral} &
\begin{minipage}{13.5em}\raggedright
{\bf{Comment}}
\end{minipage}
\\
\hline
\noalign{\vskip 4pt}

& & Schwarzschild ($\Lambda=0$) & Yes \cite{Chandrasekhar_1983} &
\begin{minipage}{13.5em}\raggedright
Satisfies Lemma~\ref{lem: ChiralLemmaParity}
\end{minipage} \\

\begin{minipage}{10em}\raggedright
General Relativity in Vacuum
\end{minipage}
& $R-2\Lambda$
& Kerr ($\Lambda=0$)
& Yes \cite{Li:2023ulk} &
\begin{minipage}{13.5em}\raggedright
Satisfies Lemma~\ref{lem: ChiralLemmaParity}
\end{minipage} \\

& & Kerr--de Sitter $(\Lambda>0)$ & Yes \cite{SupplementalMaterial} &
\begin{minipage}{13.5em}\raggedright
Satisfies Lemma~\ref{lem: ChiralLemmaParity}
\end{minipage} \\

& &
\begin{minipage}{9em}\raggedright
Kerr--Anti-de Sitter $(\Lambda<0)$
\end{minipage}
& No \cite{SupplementalMaterial} &
\begin{minipage}{13.5em}\raggedright
Theorem~\ref{thm: MainTheorem} does not apply
\end{minipage} \\

\hline
\noalign{\vskip 4pt}

&
$\kappa_gR+S_{\mathrm{shell}}$
&
\begin{minipage}{15em}\raggedright
Schwarzschild with thin shell \cite{Laeuger:2025zgb}
\end{minipage}
& No \cite{Laeuger:2025zgb} &
\begin{minipage}{13.5em}\raggedright
Theorem~\ref{thm: MainTheorem} does not apply
\end{minipage} \\
\noalign{\vskip 4pt}
&
$\kappa_gR+S_{\mathrm{M}}$
&
\begin{minipage}{13em}\raggedright
Linear, Static, Bumpy ($\Lambda=0$)
\end{minipage}
& No \cite{Weller:2024qvo} &
\begin{minipage}{13.5em}\raggedright
Theorem~\ref{thm: MainTheorem} does not apply
\end{minipage} \\

\begin{minipage}{10em}\raggedright
General Relativity with Matter
\end{minipage}
& 
& Johannsen--Psaltis \cite{Johannsen:2011dh}
& No \cite{Wu:2025obg} &
\begin{minipage}{13.5em}\raggedright
Theorem~\ref{thm: MainTheorem} does not apply
\end{minipage} \\

&
&
Membrane spacetime \cite{Thorne:1986iy}
&
No \cite{Saketh:2024ojw}
&
\begin{minipage}{13.5em}\raggedright
Theorem~\ref{thm: MainTheorem} does not apply
\end{minipage} \\

\noalign{\vskip 4pt}
\hline
\noalign{\vskip 4pt}

Einstein--Maxwell
& $\kappa_gR-F^2/4$
& Reissner--Nordstr\"om
& Yes \cite{Chandrasekhar_1983} &
\begin{minipage}{13.5em}\raggedright
Satisfies Lemma~\ref{lem: ChiralLemmaParity}
\end{minipage} \\

& & Kerr--Newman \cite{Newman:1965my}& Yes \cite{Weller:2026kk} &
\begin{minipage}{13.5em}\raggedright
Satisfies Lemma~\ref{lem: ChiralLemmaParity}
\end{minipage} \\

\hline
\noalign{\vskip 4pt}

\begin{minipage}{8em}\raggedright
Nonlinear Einstein--Maxwell
\end{minipage}
& $\kappa_gR+\mathcal{L}(F^2)$
&
\begin{minipage}{9em}\raggedright
Bardeen--de Sitter \cite{Fernando:2016ksb}
\end{minipage}
& No \cite{Zhao:2023jiz} &
\begin{minipage}{13.5em}\raggedright
Theorem~\ref{thm: MainTheorem} does not apply
\end{minipage} \\

\noalign{\vskip 4pt}
\hline
\noalign{\vskip 4pt}

\begin{minipage}{8em}\raggedright
Einstein--Maxwell--scalar
\end{minipage}
&
\begin{minipage}{12.5em}\raggedright
$\kappa_gR-2(\nabla\phi)^2-f(\phi)F_{\mu\nu}F^{\mu\nu}$
\end{minipage}
&
\begin{minipage}{9em}\raggedright
Scalarized Reissner--Nordstr\"om~\cite{Blazquez-Salcedo:2020nhs}
\end{minipage}
& No \cite{LuisBlazquez-Salcedo:2020rqp} &
\begin{minipage}{13.5em}\raggedright
Theorem~\ref{thm: MainTheorem} does not apply
\end{minipage} \\

\noalign{\vskip 4pt}
\hline
\noalign{\vskip 4pt}

\begin{minipage}[c]{8em}\raggedright
Einstein--Maxwell--dilaton--axion
\end{minipage}
&
\begin{minipage}[c]{13em}
\raggedright
\(\displaystyle
\begin{array}{@{}l@{}}
\kappa_gR-\frac{1}{2}(\nabla\phi)^2-e^\phi F^2 \\[-2pt]
-\frac{1}{2}e^{-2\phi}(\nabla\chi)^2
+b\chi\widetilde{F}^{\mu\nu}F_{\mu\nu}
\end{array}
\)
\end{minipage}
&
\begin{minipage}{9em}\raggedright
Charged Gibbons and Maeda \cite{Gibbons:1987ps}
\end{minipage}
& Yes \cite{Pope:2025jgz} &
\begin{minipage}{13.5em}\raggedright
Theorem~\ref{thm: MainTheorem} does not apply
\end{minipage} \\

\noalign{\vskip 4pt}
\hline
\noalign{\vskip 4pt}

\begin{minipage}{8em}\raggedright
Higher Order Starobinsky--Podolsky gravity
\end{minipage}
&
\begin{minipage}[c]{14em}
\raggedright
\(\displaystyle
\begin{array}{@{}l@{}}
f(R) + c_1 \nabla_{\mu}R \nabla^{\mu}R+ \dots \\[-2pt]
 + c_N\nabla_{\mu_1}...\nabla_{\mu_N}R \nabla^{\mu_1}...\nabla^{\mu_N}R
\end{array}
\)
\end{minipage}
&
Kerr ($\Lambda=0$)
& Yes (tensor sector) \cite{SupplementalMaterial} &
\begin{minipage}{13.5em}\raggedright
Satisfies Theorem~\ref{thm: TheoremB} (tensor sector)
\end{minipage} \\

\noalign{\vskip 4pt}
\hline
\noalign{\vskip 4pt}

\begin{minipage}{8em}\raggedright
Higher-Curvature Gravity
\end{minipage}
&
\begin{minipage}{12em}\raggedright
$\kappa_gR+\ell^4\lambda_{\mathrm{ev}}\mathcal{R}^3+\dots$
\end{minipage}
&
Beyond Kerr~\cite{Cano:2019ore}
& No \cite{Cano:2020cao} &
\begin{minipage}{13.5em}\raggedright
Theorem~\ref{thm: MainTheorem} does not apply
\end{minipage} \\

\noalign{\vskip 4pt}
\hline
\noalign{\vskip 4pt}

\begin{minipage}{8em}\raggedright
Einstein--dilaton--Gauss--Bonnet gravity
\end{minipage}
&
$\kappa_gR+\alpha_{\mathrm{Ed}}\mathcal{G}e^\varphi-(\nabla\varphi)^2/2$
&
\begin{minipage}{9em}\raggedright
Beyond-Kerr with monopole scalar hair~\cite{R:2022tqa,Kleihaus:2011tg,Ayzenberg:2014aka,Maselli:2015tta,Kleihaus:2015aje}
\end{minipage}
& No \cite{Li:2023ulk} &
\begin{minipage}{13.5em}\raggedright
Theorem~\ref{thm: MainTheorem} does not apply
\end{minipage} \\

\noalign{\vskip 4pt}
\hline
\noalign{\vskip 4pt}

\begin{minipage}{8em}\raggedright
Scalar--Gauss--Bonnet gravity
\end{minipage}
&
\begin{minipage}{12.5em}\raggedright
$\kappa_gR+\alpha_{\mathrm{sGB}} f(\varphi)\mathcal{G}-(\nabla\varphi)^2/2$
\end{minipage}
&
\begin{minipage}{9em}\raggedright
Beyond-Kerr black hole with monopole scalar hair~\cite{Chung:2024vaf,Cunha:2019dwb,Collodel:2019kkx,Herdeiro:2020wei,Berti:2020kgk,Dima:2020yac}
\end{minipage}
& No \cite{Chung:2024vaf} &
\begin{minipage}{13.5em}\raggedright
Theorem~\ref{thm: MainTheorem} does not apply
\end{minipage} \\

\noalign{\vskip 4pt}
\hline
\noalign{\vskip 4pt}

\begin{minipage}{8em}\raggedright
Dynamical Chern--Simons gravity
\end{minipage}
&
$\kappa_gR+\alpha_{\mathrm{dCS}}\varphi R{}^{*}R-(\nabla\varphi)^2/2$
&
\begin{minipage}{9em}\raggedright
Beyond-Kerr with dipole scalar hair~\cite{R:2022tqa,Yunes:2009hc,Yagi:2012ya,Stein:2014xba,Delsate:2018ome,Chung:2025gyg}
\end{minipage}
& No \cite{Li:2023ulk} &
\begin{minipage}{13.5em}\raggedright
Theorem~\ref{thm: MainTheorem} does not apply
\end{minipage} \\

\noalign{\vskip 4pt}
\hline
\noalign{\vskip 4pt}

\begin{minipage}[c]{8em}\raggedright
Einstein--Bumblebee gravity
\end{minipage}
&
\begin{minipage}[c]{13em}
\raggedright
\(\displaystyle
\begin{array}{@{}l@{}}
\kappa_g(R-2\Lambda)
+\zeta \kappa_gB^\mu B^\nu R_{\mu\nu} \\[-2pt]
-\frac{1}{4}B^{\mu\nu}B_{\mu\nu}
-f\left(B^2\pm b^2\right)
\end{array}
\)
\end{minipage}
&
\begin{minipage}{9em}\raggedright
Bumblebee--Schwarzschild--de Sitter \cite{Maluf:2020kgf}
\end{minipage}
& No \cite{Liu:2024oeq} &
\begin{minipage}{13.5em}\raggedright
Theorem~\ref{thm: MainTheorem} does not apply
\end{minipage}

\end{tabular}
\end{ruledtabular}
\endgroup
\end{table}
Throughout the Supplemental Material, we use the modified directional derivatives
defined in the End Matter [see Eq.~\eqref{eq: ModifiedDirectionDeriv}]. 
In the main text, we prove that Eq.~\eqref{eq: RicciFormCurved} forms two, independent, third-order differential equations for the complex metric components $\{h_{nm},h_{mm}\}$ in the ingoing radiation gauge. A similar set of equations can also be formed in the outgoing radiation gauge for $\{h_{l\bar{m}},h_{\bar{m}\bar{m}}\}$.
We recall that in deriving Eq.~\eqref{eq: RicciFormCurved} we assumed the frame gauge of Eq.~\eqref{eq: CKFrameGauge} at $\mathcal{O}(\epsilon^0)$ and of Eq.~\eqref{eq: CCKFrameIRG} at $\mathcal{O}(\epsilon^1)$. We now detail the procedure to solve for the remaining metric component $h_{nn}$. 

We begin by considering the following Ricci identity: 
\begin{align}\label{eq: Riccihnn}
    \delta_{[1,3,0,-1]} \nu-
    \Delta_{[1,0,1,1]}\mu
    -\lambda \bar{\lambda}+\bar{\nu} \pi=\Phi_{22}\,,
\end{align}
which, expanding to $\mathcal{O}(\epsilon)$, yields
\begin{align}\label{eq: RiccihnnPerturb}
    \delta_{[1,3,0,-1]} \nu^{(1)}-\left(\Delta_{[1,0,1,1]}\right)^{(1)} \mu - \left(\Delta_{[1,0,1,1]}\right) \mu^{(1)}+\bar{\nu}^{(1)} \pi=\Phi^{(1)}_{22}\,,
\end{align}
where we have used that $\nu^{(0)} = 0 = \lambda^{(0)}$. All the perturbed spin coefficients in Eq.~\eqref{eq: RiccihnnPerturb} can be expressed in metric components using the frame gauge of Eq.~\eqref{eq: CCKFrameIRG}, namely \cite{Wagle:2023fwl}
\begin{subequations}
\begin{align}
    \nu^{(1)}&= -\frac{1}{2} \bar{\delta}_{[2,2,-1,-1]} h_{n n}^{(1)}+\Delta_{[0,1,2,0]} h_{n \bar{m}}^{(1)}\,,\\
    \mu^{(1)}&=\frac{1}{2}\left[-\rho h_{n n}^{(1)}-\bar{\delta}_{[0,2,-2,-1]} h_{n m}^{(1)}+\delta_{[0,2,0,1]} h_{n \bar{m}}^{(1)}\right]\,,\\
    \gamma^{(1)}&=\frac{1}{4}\left[-D_{[0,2,1,-1]} h_{n n}^{(1)}-\bar{\delta}_{[0,2,-2,-1]} h_{n m}^{(1)}+\delta_{[0,2,2,3]} h_{n \bar{m}}^{(1)}\right],
\end{align}
\end{subequations}
The NP Ricci scalar $\Phi^{(1)}_{22}$ can be expressed in the ingoing radiation gauge as $\Phi_{22} = \frac{1}{2} (R_{\mu \nu} -\frac{1}{4} R g_{\mu \nu})n^{\mu}n^{\nu} = \frac{1}{4 \kappa_g}(T^{\mathrm{eff}}_{\mu \nu} + \frac{1}{4}Rg_{\mu \nu})n^{\nu}n^{\mu}$ so therefore,
\begin{align}
    \Phi^{(1)}_{22} =\frac{1}{4\kappa_g}\left(T^{\mathrm{eff}(1)}_{nn} + h_{nn}T^{\mathrm{eff}(0)}_{ln}\right)\,.
\end{align}
Since we have assumed $T^{\mathrm{eff}(1)}_{nn}$ vanishes by condition (iii) of Def.~\ref{def: ChiralAligned} and $\Delta^{(1)} = h_{nn}D/2 $, Eq.~\eqref{eq: RiccihnnPerturb} is a second-order transport equation for the metric component $h_{nn}$ in terms of the known metric quantities $h_{nm}$ and $h_{mm}$, which were determined from Eq.~\eqref{eq: RicciFormCurved}. A similar transport equation can be found for $h_{ll}$ if one chooses to reconstruct the metric in the outgoing radiation gauge.

\subsection{Chiral-Aligned Example: Kerr--de Sitter}
We consider the Einstein--Hilbert action with a positive cosmological constant
\begin{align}\label{eq: deSitterAction}
    S[g] =  \kappa_g\int \,d^4x \sqrt{|g|}\left( R  -2 \Lambda\right)\quad \Lambda \geq 0
\end{align}
where the equations of motion are given by varying Eq.~\eqref{eq: deSitterAction} with respect to the metric $g_{\mu \nu}$ and yield
\begin{align}\label{eq: deSitterEquOfMot}
    R_{\mu \nu}-\frac{1}{2} R g_{\mu \nu}+\Lambda g_{\mu \nu}=0\,.
\end{align}
For this example, we consider linear perturbations of the Kerr--de Sitter (KdS) black hole \cite{Carter:1968ks,Carter:1968rr}. For reviews on the KdS spacetime, see Refs.~\cite{Akcay:2010vt,Frolov:2017kze}. The KdS solution is a special case of the Pleba\'nski--Demia\'nski family,
obtained by setting the acceleration parameter, the NUT charge, and the
electric and magnetic charges to zero. In Boyer--Lindquist coordinates the  spacetime can be expressed by the following line element \cite{Konoplya:2007zx, Suzuki:1998vy}: 
\begin{align}\label{eq: KerrdS}
d s^2&=  \rho^2\left(\frac{d r^2}{\Delta_r}+\frac{d \theta^2}{\Delta_\theta}\right)+\frac{\Delta_\theta \sin ^2 \theta}{(1+\alpha)^2 \rho^2}\left[a d t-\left(r^2+a^2\right) d \varphi\right]^2  -\frac{\Delta_r}{(1+\alpha)^2 \rho^2}\left(d t-a \sin ^2 \theta d \varphi\right)^2
\end{align}
where
\begin{align}
\Delta_r=\left(r^2+a^2\right)\left(1-\alpha r^2 / a^2\right)-2 M r, \quad;\quad
\alpha=\Lambda a^2 / 3\quad; \quad \Delta_\theta=1+\alpha \cos ^2 \theta,\quad;\quad
\rho^2=r^2+a^2 \cos ^2 \theta 
\end{align}
and $M$ and $a$ are the usual mass and spin parameters of the black hole. 
It is well known that this solution is conformally K\"ahler \cite{Kubiznak:2007kh} and possesses a Killing--Yano tensor $f^{\mathrm{KY}}_{\mu \nu}$ given by
\begin{align}\label{eq: kDSCKY}
\mathbf{f}^{\rm KY} = \frac{1}{2}f^{\rm KY}_{\mu \nu} dx^{\mu} \wedge dx^{\nu}
=
a\cos\theta\,dr\wedge
\left(dt-a\sin^2\theta\,d\phi\right)
-
r\sin\theta\,d\theta\wedge
\left[a\,dt-\left(r^2+a^2\right)d\phi\right]\, 
\end{align}
where $dx^{\mu} \wedge dx^{\nu}$ denotes the usual wedge product.
Moreover, the stress-energy tensor of the spacetime $T_{\mu \nu} \propto \Lambda g_{\mu \nu}$, which implies the matter-aligned condition is also satisfied. This shows that both conditions (i) and (ii) of Def.~ \ref{def: ChiralAligned} are satisfied.

The linearized equations of motion are given by expanding Eq.~\eqref{eq: deSitterEquOfMot} to $\mathcal{O}(\epsilon)$
\begin{align}\label{eq: KdSPerturbed}
    R^{(1)}_{\mu \nu} - \frac{1}{2}h^{(1)}_{\mu \nu}R - \frac{1}{2}g_{\mu \nu}R^{(1)} + \Lambda h^{(1)}_{\mu \nu}&= 0\,.
\end{align}
Using the on-shell relation for the trace $R^{(1)} = 0$ and $R^{(1)}_{\mu \nu} = \Lambda h_{\mu \nu}$, we find the perturbed traceless Ricci tensor vanishes 
\begin{align}
    S^{(1)}_{\mu \nu} = R^{(1)}_{\mu \nu} - \frac{1}{4}g_{\mu \nu} R^{(1)} - \frac{1}{4}h_{\mu \nu}R &= 0\,.
\end{align}
It follows that all of the Ricci scalars $\Phi^{(1)}_{ij}$ and $\Lambda_{\mathrm{NP}}^{(1)}$ vanish, and the two Bianchi identities used in the main text [Eqs.~\eqref{eq: Bianchi1} and \eqref{eq: Bianchi2}] reduce to their sourceless form. Therefore, condition (iv) of Def.~\ref{def: ChiralAligned} is also satisfied. 

Linear perturbations of the KdS background satisfy the homogeneous Teukolsky equations
\begin{align}\label{eq: KerrDeSitterTek}
\mathcal{O}_{4}\Psi^{(1)}_4 =0\quad; \quad 
\mathcal{O}_{0}\Psi^{(1)}_0 =0\,.
\end{align}
This equation must be supplemented with the correct boundary conditions, which impose that the metric perturbation is purely ingoing at the future event horizon and purely outgoing at the cosmological horizon \cite{Yoshida:2010zzb}.  This boundary condition is similar to that associated with perturbations of the Kerr background, which is also complex-linear and imposes the same boundary condition on the even- and odd-parity components of the metric perturbation.  

We can also verify that there is a metric reconstruction procedure by imposing the traceless ingoing radiation gauge. After fixing $h_{\mu \nu}l^{\mu} = 0$, the effective perturbed stress-energy tensor component $T^{\mathrm{eff}(1)}_{\mu \nu}l^{\mu}l^{\nu} = \Lambda h_{ll}$ vanishes. This satisfies condition (iii) of Def.~\ref{def: ChiralAligned}, and allows us to set the trace of the metric perturbation to zero with residual gauge \cite{Price:2006ke}. With all the conditions of Def.~\ref{def: ChiralAligned} now satisfied we are guaranteed that a metric reconstruction procedure exists by Lemma \ref{lem: ChiralLemma}, because $\{h^{(1)}_{nm},h^{(1)}_{mm}\}$ can be solved for with Eq.~\eqref{eq: RicciFormCurved} and $h^{(1)}_{nn}$ can be solved for using Eq.~\eqref{eq: RiccihnnPerturb}. Since we have shown that the boundary conditions are also complex-linear and the KdS spacetime admits a tetrad satsifying $\hat{P}m^{\mu} = - \bar{m}^{\mu}$, Lemma \ref{lem: ChiralLemmaParity} applies in this case, and the linear perturbations of KdS must be parity isospectral.

As a side note, one may also implement other reconstruction procedures, such as those given in~\cite{Suvorov:2019qow, Loutrel:2020wbw, Ripley:2020xby} or more recently in~\cite{Li:2026rkf}, which do not require solving for $h^{(1)}_{nm}$ and $h^{(1)}_{mm}$ simultaneously. A similar procedure to the Chrzanowski--Cohen--Kegeles reconstruction was given in~\cite{Berens:2025kkm}. As opposed to other methods that directly integrate a subset of the Ricci and Bianchi identities systematically, the Chrzanowski--Cohen--Kegeles method leverages the self-adjoint property of the linearized Einstein operator and the real metric perturbation can be written as a second-order operator on an intermediate Hertz potential \cite{Wald:1978vm}. 

In contradistinction, linear perturbations of Kerr--AdS are \textit{not} isospectral. While the homogeneous Teukolsky equations are still satisfied and the background spacetime is conformally K\"ahler and matter aligned, the boundary conditions spoil the isospectral relation found for Kerr and KdS in GR. The boundary condition at the future event horizon is unmodified, but the outgoing wave condition at the cosmological horizon (or future null infinity) is replaced with the requirement that the metric perturbation preserve the global AdS symmetry group $O(3,2)$ asymptotically \cite{Henneaux:1985tv, Dias:2013sdc}. Writing the two power-law behaviors of the
radial Teukolsky functions as 
\begin{subequations}
\begin{align}
{}_{+2}R_{\ell m\omega }\bigg|_{r \rightarrow \infty} &\sim A_{+}^{(2)} \frac{L}{r}+A_{-}^{(2)} \frac{L^2}{r^2}+\mathcal{O}\left(\frac{L^3}{r^3}\right)\,, \\
 {}_{-2}R_{\ell m\omega }\bigg|_{r \rightarrow \infty} &\sim B_{+}^{(-2)} \frac{L}{r}+B_{-}^{(-2)} \frac{L^2}{r^2}+\mathcal{O}\left(\frac{L^3}{r^3}\right)\,,
\end{align}
\end{subequations}
where $L$ is the characteristic AdS curvature, and $A^{(\pm2)}$ and $B^{(\pm 2)}$ are complex-valued mode amplitudes. Ref.~\cite{Dias:2013sdc} finds the asymptotic behavior is satisfied if and only if $A_{-}^{(2)}=-i \eta A_{+}^{(2)}$ and $ B_{-}^{(-2)}=i \eta B_{+}^{(-2)}$, where $\eta$ can take two distinct values: $\eta \in \{\eta_{\rm s},\eta_{\rm v}\}$ 
with $\eta_{\rm s}\neq\eta_{\rm v}$. In the nonrotating limit, $\eta_s$ and $\eta_v$ reduce to the boundary conditions of the
even-parity and odd-parity sectors, respectively
\cite{Dias:2013sdc}. The boundary condition is therefore parity dependent:
each branch defines a distinct eigenvalue problem, the two sectors no longer
obey the same complex-linear boundary conditions, and the spectra split
\cite{Cardoso:2001bb,Berti:2003ud}.

An apparent tension arises with Ref.~\cite{Tattersall:2018axd}, which reports a
splitting of the even- and odd-parity QNM frequencies of slowly-rotating
Kerr--de Sitter, when both $a$ and $\Lambda$ are nonzero. That result was obtained
for the fundamental mode from master equations truncated at $\mathcal{O}(a)$ and
solved with the inverse-multipole expansion of Ref.~\cite{Dolan:2009nk}. The two
sectors agree through $\mathcal{O}(L^{-3})$ with $L=\ell+1/2$; the first
difference enters at $\mathcal{O}(L^{-4})$ and is proportional to $am\Lambda M$,
so it vanishes for $m=0$ and in both the $a\to0$ and $\Lambda\to0$ limits. The
expansion of Ref.~\cite{Dolan:2009nk} encodes the QNM boundary conditions in an
ansatz whose phase is fixed by the critical null geodesics of a \emph{static},
spherically symmetric metric, imposing $u\sim e^{\mp i\omega r_{*}}$ at both
boundaries. However, on a rotating background the correct conditions are co-rotating,
$u\sim e^{- i(\omega-m\Omega_{H})r_{*}}$ at the future event horizon and $u\sim e^{+ i(\omega-m\Omega_{c})r_{*}}$ at the cosmological horizon, where $\Omega_{H} = a/r^{2}_{H}$ and $\Omega_{c} = a/r^{2}_{c}$ are their respective angular velocities up to $\mathcal{O}(a)$ and $(r_H,r_c)$ are the two positive roots of $1 - 2M/r - \Lambda r^2/3 = 0$ . We therefore interpret the reported
splitting as an artifact of the approximation scheme rather than a breakdown of parity
isospectrality.

\subsection{Chiral-Aligned Example: Kerr--Newman in Einstein--Maxwell}

For this example, we consider the Einstein--Maxwell action
\begin{align}
    S[g, A]=\int \mathrm{d}^4 x \sqrt{|g|}\left(\kappa_g R-\frac{1}{4} F_{\mu \nu} F^{\mu \nu}\right)
\end{align}
where the Maxwell tensor has the standard definition $ F_{\mu \nu}=2 \nabla_{[\mu} A_{\nu]}$ and $A_{\mu}$ is the spin-1 gauge field.
In Boyer-Lindquist coordinates, the Kerr--Newman black hole is given by the following line element \cite{Konoplya:2007zx}:
\begin{equation}\label{eq:KNdsMetric}
        d s^2=  \rho^2\left(\frac{d r^2}{\Delta_r}+d \theta^2\right)+\frac{ \sin ^2 \theta}{ \rho^2}\left[a d t-\left(r^2+a^2\right) d \varphi\right]^2  -\frac{\Delta_r}{ \rho^2}\left(d t-a \sin ^2 \theta d \varphi\right)^2
    \end{equation}
where  
\begin{align}
\Delta_r=\left(r^2+a^2\right)-2 M r+Q^2,
 \quad 
\rho^2=r^2+a^2 \cos ^2 \theta 
\end{align}
This solution is also conformally K\"ahler, it admits the same Killing--Yano tensor as the KdS metric [see Eq.~\eqref{eq: kDSCKY}], and is matter-aligned because the only non-vanishing NP Ricci scalars are $\Phi_{11}$.

The operators appearing in Eq.~\eqref{eq: KNEqus} are given by \cite{Dias:2015wqa}
\begin{subequations}
\label{eq: KNModifiedOperators}
\begin{align}
\mathcal{O}_{-2}
&=
\Delta_{[4,1,3,-1]}D_{[0,0,-1,0]}
-\bar{\delta}_{[3,1,4,-1]}\delta_{[0,4,0,-1]}
-3\Psi_2\,,\\
\mathcal{P}_{-2}
&=2
-4\Delta_{[0,0,3,-1]}A_{-}D_{[0,0,-1,0]}
-4\left(\bar{\tau}-\pi\right)A_{+}
\left(\bar{\delta}+4\beta-\tau\right)\,,\\
\mathcal{Q}_{-2}
&=\frac{2}{\Phi_1^{*(0)}}\Bigl\{\Delta_{[2,0,3,-1]}A_{-}\bar{\delta}_{[2,0,6,0]}+\left(\bar{\tau}-\pi\right)A_{+}
\Delta_{[6,0,2,0]}
\Bigr\}\,,\\
\mathcal{O}_{-1}
&=
\Delta_{[5,1,3,1]}D_{[0,0,-4,0]}
-\bar{\delta}_{[1,1,5,-1]}\delta_{[0,2,0,-4]}\,,\\
\mathcal{P}_{-1}
&=2D_{[0,0,-4,1]}A_{+}\Delta_{[6,0,2,0]}+2\delta_{[-1,3,-1,-4]}A_{-}
\bar{\delta}_{[2,0,6,0]}\, ,\\
\mathcal{Q}_{-1}
&=-4\Phi_1^{(0)}\Bigl\{
D_{[0,0,-2,1]}A_{+}
\left(\bar{\delta}+4\beta-\tau\right)
+\delta_{[-1,3,1,-2]}A_{-}
\left(\Delta_{[0,0,4,0]}-\rho\right)
\Bigr\} 
\end{align}
\end{subequations}
where $A_{ \pm}=\left(3 \Psi_2^{(0)} \mp 2 \Phi_{11}^{(0)}\right)^{-1}$.
In a subsequent publication \cite{Weller:2026kk} we show that from the solutions of Eq.~\eqref{eq: KNEqus} one can reconstruct the metric perturbation $h^{(1)}_{\mu \nu}$ and spin-1 gauge field $A^{(1)}_{\mu}$ in an ingoing or outgoing radiation gauge. To compute QNMs, a similar set of complex-linear boundary conditions are imposed as in the Kerr case: the entire gravitoelectromagnetic perturbation is purely ingoing at the future horizon and outgoing at future null infinity. Moreover, we will show that all the conditions of Def.~\ref{def: ChiralAligned} are satisfied, the boundary conditions are complex-linear, and the background tetrad satisfies the hypothesis of Lemma \ref{lem: ChiralLemmaParity}, therefore the linear perturbations must be parity isospectral. 

\subsection{Non-Chiral Aligned Example: Higher Order Starobinsky--Podolsky gravity}

In this example, we consider a subset of $f\left(R, \nabla_{\mu_1} R,\ldots, \nabla_{\mu_1} \ldots \nabla_{\mu_n} R\right)$ gravity.  The action is given by
\begin{align}\label{eq: StarobinskyAction}
S[g] & =  \int d^4 x \sqrt{|g|} \left(f(R) + c_1 \nabla_{\mu}R \nabla^{\mu}R + \dots + c_N\nabla_{\mu_1}...\nabla_{\mu_N}R \nabla^{\mu_1}...\nabla^{\mu_N}R \right) 
\end{align}
where $g$ denotes the Lorentzian metric and $N$ is a positive integer. We consider linear perturbations about the Kerr black hole in the gravitational theory defined by the action in Eq.~\eqref{eq: StarobinskyAction}. The Kerr black hole is still a solution at $\mathcal{O}(\epsilon^0)$, since the scalar curvature along with its corresponding derivatives all vanish identically. 

We assume that $f(R)$ can be expanded in the following Maclaurin series 
\begin{align}\label{eq: MaclSeries}
    f(R) = R + \frac{c_0}{2} R^2 + \dots\,.
\end{align}
If we consider $c_i = 0$ for $i\geq 1$, the equations of motion are given by 
\begin{align}\label{eq: f(R)EquOfMotion}
    f_R R_{\mu \nu}-\frac{1}{2} f g_{\mu \nu}+\left(g_{\mu \nu} \square-\nabla_\mu \nabla_\nu\right) f_R = 0\,, 
\end{align}
where $f_R \equiv df/dR$. Expanding Eq.~\eqref{eq: f(R)EquOfMotion} to $\mathcal{O}(\epsilon^1)$, we find higher-order terms from the Maclaurin series in Eq.~\eqref{eq: MaclSeries} vanish as the background solution is Kerr. The equations of motion at $\mathcal{O}(\epsilon^1)$, assuming $c_i = 0$ for $i\geq 1$, are given by
\begin{align}
    R^{(1)}_{\mu \nu}-\frac{1}{2} R^{(1)}g_{\mu \nu} + c_0  \left(g_{\mu \nu} \square - \nabla_{\mu}\nabla_{\nu}\right) R^{(1)} = 0\,.
\end{align}

If we consider nonvanishing $c_i$, the equations of motion receive a correction from the following term
\begin{align}
    \frac{\delta }{\delta g^{\mu \nu}}\left(\sqrt{|g|}c_N\nabla_{\mu_1}...\nabla_{\mu_N}R \nabla^{\mu_1}...\nabla^{\mu_N}R\right)
\end{align}
Upon linearization up to $\mathcal{O}(\epsilon^1)$, the only nonvanishing contribution comes from
\begin{align}
    \left(\sqrt{|g|}c_N\nabla_{\mu_1}...\nabla_{\mu_N}\frac{\delta R}{\delta g^{\mu \nu}} \nabla^{\mu_1}...\nabla^{\mu_N}R\right)\,.
\end{align}
Other terms appearing in the equations of motion are quadratic in derivatives of the background scalar curvature and thus vanish at $\mathcal{O}(\epsilon^1)$ when Kerr is the chosen solution at $\mathcal{O}(\epsilon^0)$. 
Thus, the linear equations of motion are given by
\begin{align}\label{eq: HigherDerivEquOfMot}
    G_{\mu \nu}^{(1)}[h]+\left(g_{\mu \nu} \square-\nabla_\mu \nabla_\nu\right) \mathcal{P}_N R^{(1)}[h]=0 \,,
\end{align}
where 
\begin{subequations}
\begin{align}
\mathcal{O}_j & \equiv(-1)^j \nabla_{\mu_1} \cdots \nabla_{\mu_j} \nabla^{\mu_j} \cdots \nabla^{\mu_1}\,, \quad j \geq 1 \\
\mathcal{P}_N & \equiv c_0 + 2\sum_{j=1}^N c_j \mathcal{O}_j\,,\quad N\geq 0
\end{align}
\end{subequations}
 Taking the trace of Eq.~\eqref{eq: HigherDerivEquOfMot} yields a decoupled partial differential equation for $R^{(1)}[h]$
\begin{align}\label{eq: HigherDerivTrace}
     \left(3\square \mathcal{P}_N - 1\right)R^{(1)}[h]=0 
\end{align}
If we set $c_j = 0$ for all $j$ then the linearized equations reduce to vacuum perturbations of Kerr where $G^{(1)}_{\mu \nu}[h] = 0$, as expected.

We now demonstrate that the projection of Eq.~\eqref{eq: HigherDerivEquOfMot} onto the Newman--Penrose basis yields the homogeneous Teukolsky equations for $\Psi^{(1)}_{0,4}$. Consider a metric $\widetilde{g}_{\mu \nu}$ which is a conformal transformation of a metric $g_{\mu \nu}$,  
\begin{align}
    \widetilde{g}_{\mu \nu} = \Omega^2g_{\mu \nu}\,,
\end{align} 
    where $\Omega$, known as the conformal factor, is a scalar function of the coordinates. Under a conformal transformation the Weyl tensor transforms as $W_{\mu \nu \gamma \delta} \mapsto \Omega^2 W_{\mu \nu \gamma \delta}$, while the Weyl scalars transform as $\Psi_i \mapsto \Omega^{-2}\Psi_i $, 
    where we have assumed that the tetrad is scaled linearly with $\Omega$, so the corresponding null tetrad $\{\widetilde{e}_a\}$, corresponding to $\tilde{g}_{\mu \nu}$, satisfies $\{\widetilde{e}^{\mu}_{a}\} = \{\Omega^{-1} e^{\mu}_{a}\}$ (for other choices see \cite{Penrose:1986ca}). Suppose the Ricci tensor generated by $g_{\mu \nu}$ is $R_{\mu \nu}$. Then, the Ricci tensor associated with $\widetilde{g}_{\mu \nu}$ is given by \cite{Wald:1984cw}
\begin{align}\label{eq: ConformalRicci}
\tilde{R}_{\mu \nu}= & R_{\mu \nu}-2 \nabla_\mu \nabla_\nu \log \Omega-g_{\mu \nu} \square \log \Omega +2\left(\nabla_\mu \log \Omega\right) \nabla_\nu \log \Omega-2 g_{\mu \nu} g^{\alpha \beta }\left(\nabla_\alpha \log \Omega\right) \left(\nabla_\beta \log \Omega\right)
\end{align}
where we have evaluated Eq.~\eqref{eq: ConformalRicci} in four dimensions. Observe that the first two terms are linear in $\mathcal{O}(\log \Omega)$, while the last two are $\mathcal{O}((\log \Omega)^2)$. Therefore, the tracefree Ricci tensor for the metric perturbation $h^{\mathfrak{s}}_{\mu \nu} = 2\varphi g_{\mu \nu}$ is given by
\begin{align}\label{eq: conformalTraceFreeStress}
    S^{\mathfrak{s}}_{\mu \nu} = R^{\mathfrak{s}}_{\mu \nu} - \frac{1}{4}g_{\mu \nu}R^{\mathfrak{s}}= - 2 \nabla_{\mu}\nabla_{\nu} \varphi + \frac{1}{2}g_{\mu \nu} \square \varphi
\end{align}
where we have identified $\Omega^2 = e^{2 \varphi}$ and only kept terms up to $\mathcal{O}(\varphi^1)$. Setting $\varphi = ({1}/{2}) \mathcal{P}_N R^{(1)}[h]$, Eq.~\eqref{eq: HigherDerivEquOfMot} becomes
\begin{align}\label{eq: effectiveScalarStress}
    T^{\mathrm{eff}}_{\mu \nu} = -2\left(g_{\mu \nu}\square - \nabla_\mu \nabla_{\nu }\right)\varphi \quad ;\quad S^{\mathrm{eff}}_{\mu \nu} = -\left(\frac{1}{2}g_{\mu \nu}\square -2\nabla_\mu \nabla_{\nu }\right)\varphi
\end{align}
Since the traceless component of the  effective stress in Eq.~\eqref{eq: effectiveScalarStress}
is proportional to the tracefree stress of a conformal transformation in Eq.~\eqref{eq: conformalTraceFreeStress}, we can identify the metric perturbation that generates this source as a linearized conformal transformation. This requires that the linear perturbations of the extreme-spin Weyl scalars vanish, which means the effective Teukolsky source for the stress in Eq.~\eqref{eq: effectiveScalarStress} must also vanish,
\begin{subequations}
\begin{align}
\mathscr{T}^{(1)}_0[S^{\mathrm{eff}}_{\mu \nu}] &= 0\,,\\
\mathscr{T}^{(1)}_4[S^{\mathrm{eff}}_{\mu \nu}] &= 0\,,
\end{align}
\end{subequations}
and the homogeneous Teukolsky equations are also satisfied. 

As suggested in Ref.~\cite{Suvorov:2019qow}, the cancellation of these source terms can also be verified directly, by repeatedly applying the
equations of motion together with the commutator relations of the directional derivatives. Identifying the linearized stress--energy
tensor as a conformal transformation of the background metric,
however, makes the cancellation manifest: the effective Teukolsky source, i.e., the map $T_{\mu\nu} \mapsto \mathscr{T}^{(1)}_{4,0}$, annihilates every metric perturbation $h^{\mathfrak{s}}_{\mu\nu} = 2\varphi g_{\mu\nu}$, which is a pure-trace mode. Even though the homogeneous Teukolsky equations are still satisfied, the reconstructed metric, which corresponds to a solution of Eq.~\eqref{eq: HigherDerivEquOfMot}, does not coincide with a vacuum perturbation of Kerr. 

The trace encodes the massive Ricci degree of freedom [see Eq.~\eqref{eq: HigherDerivTrace}]. In the nonrotating limit,
where the background reduces to Schwarzschild, this scalar decomposes into even-parity harmonics alone. Because the mode has no odd-parity
counterpart, isospectrality is, strictly speaking, broken. The frequency content of the tensor
polarizations, however, is unaffected. Both $h_+$ and $h_\times$ are
reconstructed from the Weyl scalars $\Psi^{(1)}_{0,4}$, which satisfy
the homogeneous Teukolsky equations; in the ingoing radiation gauge,
for instance, $\Psi^{(1)}_0$ depends only on derivatives of $h_{mm}$. The QNM frequencies of the massive scalar mode therefore never appear
in $h_+$ or $h_\times$. This decoupling is not generic: in other
theories, such as Horndeski gravity, scalar-led QNM frequencies do appear in the tensor
polarizations \cite{Tattersall:2018nve}.

We now briefly detail how to solve for the additional massive Ricci degree of freedom. Since generically $\Phi^{(1)}_{00}$ is nonvanishing, there is no longer enough residual gauge freedom to fully implement the traceless ingoing radiation gauge. Since Eq.~\eqref{eq: StarobinskyAction} still preserves diffeomorphism symmetry, we can impose $h_{\mu \nu}l^{\mu} = 0$ by explicitly solving for the gauge vector $\xi_{\mu}$, such that
\begin{align}\label{eq: IRGGaugeTransform}
    \left(h_{\mu \nu} - \nabla_{\mu} \xi_{\nu} - \nabla_{\nu} \xi_{\mu} \right)l^{\mu} = 0
\end{align}
which yields four real first-order equations. The remaining non-vanishing metric components are given by $\{h_{nm},h_{mm},h_{nn},h_{m\bar{m}}\}$, which encode six degrees of freedom. In the ingoing radiation gauge, we can consider the following Ricci identity 
\begin{align}\label{eq: RiccillComp}
D \rho-\delta^* \kappa= & \left(\rho^2+\sigma \sigma^*\right)+\left(\varepsilon+\varepsilon^*\right) \rho-\kappa^* \tau -\kappa\left(3 \alpha+\beta^*-\pi\right)+\Phi_{00}\,,
\end{align}
which to $\mathcal{O}(\epsilon^1)$ yields
\begin{align}\label{eq: perturbedTraceTransport}
    D \rho^{(1)} = 2\rho^{(1)} \rho + (\epsilon^{(1)} + \bar{\epsilon}^{(1)})\rho + \Phi^{(1)}_{00}
\end{align}
Since $\epsilon^{(1)}$ can be expressed as
\begin{align}
    \epsilon^{(1)}=\frac{1}{4}\left[\Delta_{[-1,1,0,-2]} h_{l l}^{(1)}-2 D_{\left[0,0, \frac{1}{2},-\frac{1}{2}\right]} h_{l n}^{(1)}-\bar{\delta}_{[-2,0,-3,-2]} h_{l m}^{(1)}+\delta_{[-2,0,1,2]} h_{l \bar{m}}^{(1)}-(\rho-\bar{\rho}) h_{m \bar{m}}^{(1)}\right]\,,
\end{align}
and $\Phi^{(1)}_{00} = \frac{1}{2}D^2 \mathcal{P}_N R^{(1)}$, Eq.~\eqref{eq: perturbedTraceTransport} becomes
\begin{align}
    D \rho^{(1)} = 2\rho^{(1)} \rho + \frac{1}{2}D^2 \mathcal{P}_N R^{(1)}
\end{align}
Now substituting
\begin{align}
    \rho^{(1)}=\frac{1}{2}\left[-\mu h_{l l}^{(1)}-(\rho-\bar{\rho}) h_{l n}^{(1)}-\bar{\delta}_{[-2,0,-1,0]} h_{l m}^{(1)}+\delta_{[-2,0,1,2]} h_{l \bar{m}}^{(1)}-D_{[0,0,1,-1]} h_{m \bar{m}}^{(1)}\right]\,,
\end{align}
we have that
\begin{align}
D_{[0,0,-2,0]}D_{[0,0,1,-1]}h^{(1)}_{m\bar{m}} = D^2 \mathcal{P}_N R^{(1)}
\end{align}
which forms a second-order partial differential equation for the trace $h_{m\bar{m}}$ where the trace $R^{(1)}$ has been determined from Eq.~\eqref{eq: HigherDerivTrace}. From here we can follow the procedure given in Refs.~\cite{Suvorov:2019qow,Ripley:2020xby,Li:2026rkf}.

\end{document}